\documentclass[10pt,reqno]{pspum-l}
\usepackage{hyperref}
\usepackage{float}
\usepackage{mathtools}
\usepackage[all]{xy}

\usepackage{tikz-cd}
\usepackage{amscd,amssymb,amsmath,latexsym,bm}
\usepackage[mathcal,mathscr]{euscript}
\usepackage{lipsum}
\usepackage{amsfonts}
\usepackage{graphicx}
\usepackage{epstopdf}
\usepackage{algorithmic}
\usepackage{hyperref}
\usepackage{xcolor}
\usepackage{upgreek}
\usepackage{enumerate}
\usepackage{ulem}
\usepackage{stmaryrd}
\usepackage{gensymb}			
\usepackage[utf8]{inputenc}
\usepackage{yhmath}

\newtheorem{theorem}{Theorem}[section]
\newtheorem{lemma}[theorem]{Lemma}
\newtheorem{corollary}[theorem]{Corollary}
\newtheorem{proposition}[theorem]{Proposition}
\newtheorem{remark}[theorem]{Remark}
\newtheorem{definition}[theorem]{Definition}

\newtheorem{notation}[theorem]{Notation}

\numberwithin{equation}{section}
\numberwithin{figure}{section}

\newcommand{\BM}{{\mathbb B}}
\newcommand{\CM}{{\mathbb C}}
\newcommand{\NM}{{\mathbb N}}
\newcommand{\PM}{{\mathbb P}}

\newcommand{\RM}{{\mathbb R}}
\newcommand{\SM}{{\mathbb S}}

\newcommand{\ZM}{{\mathbb Z}}

\newcommand{\KM}{{\mathbb K}}

\newcommand{\Aa}{{\mathcal A}}

\newcommand{\Bb}{{\mathcal B}}
\newcommand{\Dd}{{\mathcal D}}
\newcommand{\Ff}{{\mathcal F}}
\newcommand{\Gg}{{\mathcal G}}

\newcommand{\Ss}{{\mathcal S}}

\newcommand{\Tt}{{\mathcal T}}

\newcommand{\Nn}{{\mathcal N}}
\newcommand{\Mm}{{\mathcal M}}
\newcommand{\Cc}{{\mathcal C}}

\newcommand{\Ll}{{\mathcal L}}
\newcommand{\Qq}{{\mathcal Q}}
\newcommand{\Kk}{{\mathcal K}}
\newcommand{\Hh}{{\mathcal H}}

\begin{document}
\title[Cyclic cocycles and quantized pairings in materials science]{Cyclic cocycles and quantized pairings in materials science}

\author{Emil Prodan}

\address{Department of Physics and
\\ Department of Mathematical Sciences 
\\Yeshiva University 
\\New York, NY 10016, USA}

\email{prodan@yu.edu}

\date{\today}

\thanks{This work was supported by the U.S. National Science Foundation through the grants DMR-1823800 and CMMI-2131760.}

\begin{abstract} 
The pairings between the cyclic cohomologies and the K-theories of separable $C^\ast$-algebras supply topological invariants that often relate to physical response coefficients of materials. Using three numerical simulations, we exemplify how some of these invariants survive throughout the full Sobolev domains of the cocycles. These interesting phenomena, which can be explained by index theorems derived from Alain Connes' quantized calculus, are now well understood in the independent electron picture. Here, we review recent developments addressing the dynamics of correlated many-fermions systems, obtained in collaboration with Bram Mesland. They supply a complete characterization of an algebra of relevant derivations over the $C^\ast$-algebra of canonical anti-commutation relations indexed by a generic discrete Delone lattice. It is argued here that these results already supply the means to generate interesting and relevant states over this algebra of derivations and to identify the cyclic cocycles corresponding to the transport coefficients of the many-fermion systems. The existing index theorems for the pairings of these cocycles, in the restrictive single fermion setting, are reviewed and updated with an emphasis on pushing the analysis on Sobolev domains. An assessment of possible generalizations to the many-body setting is given. 
\end{abstract}

\maketitle

\setcounter{tocdepth}{2}

{\scriptsize \tableofcontents}

\section{Introduction, motivations and challenges}
\label{Sec:Introduction}

Cyclic cocycles and their pairings with the K-theories of underlining algebras of macroscopic physical observables can be often related to physical response coefficients quantifying the reactions of physical systems to external stimuli. For example, the cyclic cocycles generated from circle actions  and invariant traces often relate to transport coefficients \cite{ProdanSpringer2016}, which quantify the relation between the current-density of particle ensembles and macroscopic potential-gradients across samples. One immediate value of this connection to the physicists comes from the fact that disorder can be incorporated in the analysis with virtually no extra effort. For example, the dynamics of un-correlated electrons in a perfect crystal is governed by Hamiltonians drawn from the $C^\ast$-algebra $C^\ast(\ZM^d)$ (its stabilization, to be more precise), where $d$ is the space dimension of the sample. Including disorder amounts to promoting the group $C^\ast$-algebra to a crossed product algebra $C(\Xi)\rtimes \ZM^d$, where $\Xi$ is the space of disorder configurations.

\vspace{0.2cm}

Noncommutative Geometry \cite{ConnesBook} supplies a framework where physics applications can be developed and, most importantly, where difficult calculations can be completed, {\it e.g} by following existing general templates. One such example is the charge transport theory by Bellissard and his collaborators \cite{Bellissard2003,BellissardLNP2003}, which was developed from first principles and culminated with expressions of the transport coefficients that rigorously incorporate the role of dissipation. In particular, with an expression for the Hall conductance at low temperatures and in the presence of disorder. Physicists often look at this expression through a left regular representation $\pi_\xi$ of the algebra $C(\Xi)\rtimes \ZM^2$ on $\ell^2(\ZM^2)$, which up to a multiplicative factor reads as
\begin{equation}\label{Eq:SigmaH1}
\sigma_H = \lim_{V \to \RM^2} \tfrac{1}{|V|} \sum_{x \in \ZM^2 \cap V} \langle x | P_F [[X_1,P_F],[X_2,P_F]]|x \rangle,
\end{equation}
where $P_F = \pi_\xi(p_F)$ is the Fermi projection corresponding to the Fermi energy $E_F$, i.e. the spectral projection of the Hamiltonian onto $[-\infty,E_F]$, and $(X_1,X_2)$ is the position operator on $\ell^2(\ZM^2)$. In the $C^\ast$-language, this expression reads
\begin{equation}\label{Eq:SigmaH2}
\sigma_H = \sum_{\lambda \in \Ss_2} (-1)^\lambda \Tt \big(p_F(\partial_{\lambda_1}p_F) (\partial_{\lambda_2}p_F)\big ),
\end{equation}
where $\Ss_N$ denotes the permutation group of $N$ objects and $(\Tt,\partial)$ is the standard differential structure on $C(\Xi) \rtimes \ZM^2$, with the trace $\Tt$ induced from a measure on $\Xi$. One immediately recognizes on the right side of \eqref{Eq:SigmaH2} a pairing of a cyclic cocycle with $p_F$, hence one can understand on the spot the topological content of \eqref{Eq:SigmaH2}. However, this simplicity is deceiving because we are dealing here with elements from the weak closure of the $C^\ast$-algebra, unless the Fermi energy falls in a spectral gap. This is the central theme of the applications we will discuss here, which are first developed in a $C^\ast$-algebraic context and then pushed into Sobolev domains.

\vspace{0.2cm}

For us, expression~\eqref{Eq:SigmaH2} also supplied the key to a very efficient numerical algorithm to evaluate expressions like \eqref{Eq:SigmaH1}. The difficulty in the numerical evaluation of \eqref{Eq:SigmaH1} comes from the incompatibility between the position operator and the periodic boundary conditions assumed in the computer simulations of finite-samples. The solution presented itself once we chose to work directly with the $C^\ast$-algebra rather than a representation \cite{ProdanARMA2013,ProdanSpringer2017} and, to give a flavor of what the $C^\ast$-formalism offered, we present below the optimal solution, which consists of the replacement
\begin{equation}
[X_j,A] \to  \sum_{z^{2L_j+1}=1} c_z z^{X_j} A z^{-X_j}, \quad c_z = \left \{ \begin{array}{ll}
\frac{z^{L_j}}{1-z}, \ & z \neq 1\\
0, & z =1
\end{array}
\right .
\end{equation}
where $2L_j+1$ is the size of the sample in direction $j$. Under mild assumption on the Hamiltonians, the above recipe leads to numerical algorithms that converge exponentially fast to the thermodynamic limit \cite{ProdanSpringer2017}. This gave us the opportunity, for the first time, to get close enough to the topological phase transitions, mentioned next, and compute related critical exponents and other physical quantities of interest. These ideas crystallized around the year 2009 and the first published results can be found in \cite{ProdanPRL2010}. After that, the workstations in my computational physics lab worked uninterrupted for years, on mapping the remarkable properties of the cyclic cocycle pairings. The best way to communicate what I mean by ``remarkable'' is to showcase a few examples.

\vspace{0.2cm}

Let us first examine the disordered Haldane model Hamiltonian \cite{HaldanePRL1988}
\begin{equation}\label{Eq:Model1}
H_\xi = \sum_{\langle x,y \rangle} |x\rangle \langle y| +0.6 \imath \sum_{\langle \langle x,y\rangle \rangle } \eta_x \big (|x\rangle \langle y | - |y\rangle \langle x | \big ) + W \sum_{x} \xi_x |x\rangle \langle x|,
\end{equation}
defined over the honeycomb lattice and where $\langle \cdot,\cdot \rangle / \langle\langle \cdot,\cdot \rangle \rangle$ indicate first/second near-neighboring sites and $\xi_x$ are drawn randomly and independently from the interval $[-1/2,1/2]$ (see \cite{ProdanJPA2011} for the rest of the details). $W$ will be referred to as the disorder strength. $H_\xi$ is a left regular representation of an element from $M_2(\CM) \otimes C\big ([-\tfrac{1}{2},\tfrac{1}{2}]^{\ZM^2}\big ) \rtimes \ZM^2$ and the phenomena exhibited here are generic, so the particular expression of this element is really not that important, though one needs to make sure that the Fermi projection falls into a non-trivial K-theoretic class when $W=0$ (the constant $0.6$ in Eq.~\eqref{Eq:Model1} assures just that). 

\begin{figure}
\center
  \includegraphics[width=\textwidth]{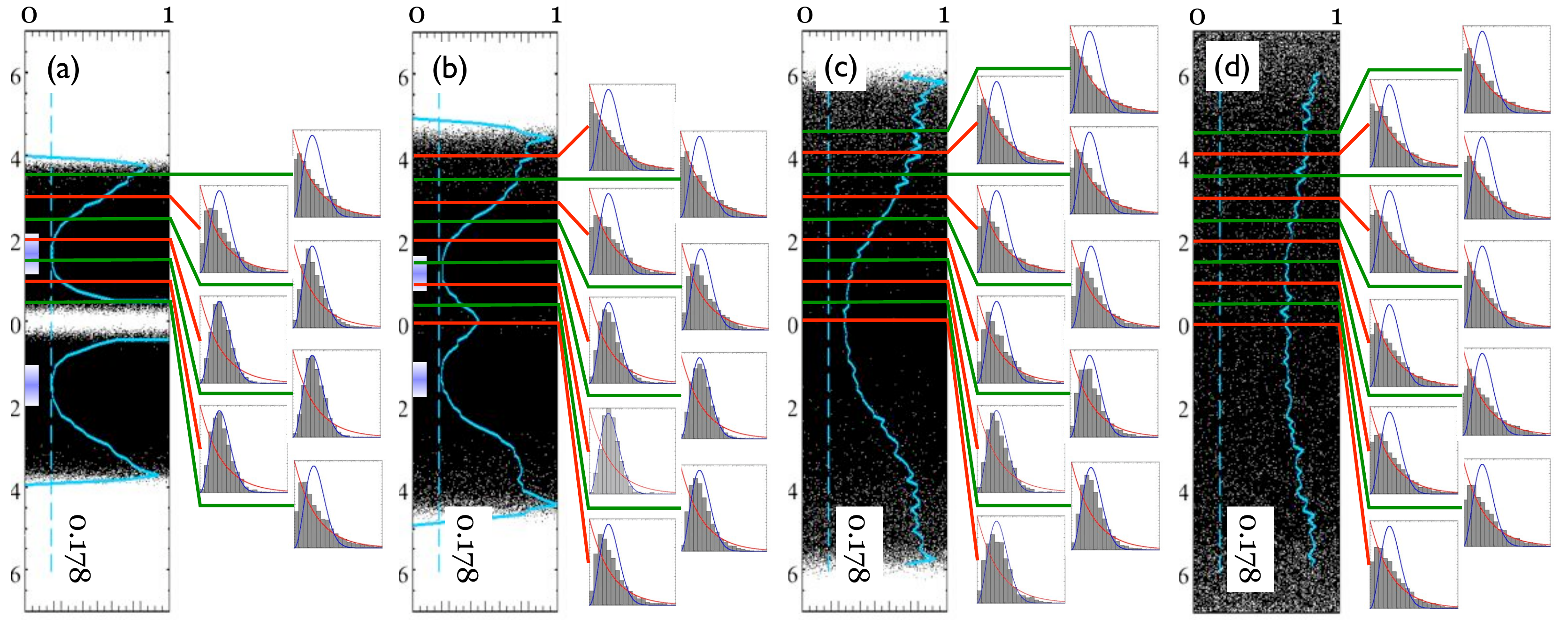}\\
  \caption{\small (Reproduced with permission from \cite{ProdanJPA2011}) Spectrum of $H_\xi$ from~\eqref{Eq:Model1}, sampled at many random configurations, together with statistical analysis indicating regions of localized and delocalized spectrum (see \cite{ProdanJPA2011} for full details). From left to right, the four panels correspond to disorder strengths $W=3$, 5, 8 and 11, respectively. Note the vanishing of the spectral gap, which occurs around $W=4$.}
\label{Fig:Fig1}
\end{figure} 

\vspace{0.2cm}

Fig.~\ref{Fig:Fig1} reports the spectrum of $H_\xi$ at various levels of disorder, evaluated for a large but finite sample. Let us point out that the spectral gap, still visible for $W=3$, vanishes at larger disorder strengths. Fig.~\ref{Fig:Fig1} also reports statistical analysis that detects the spatial character of the eigenvectors corresponding to the different parts of the spectrum. Specifically, at the highlighted points where the blue curves in Fig.~\ref{Fig:Fig1} touch the value 0.178, the eigenvectors are uniformly spread throughout the sample, while for the rest of the spectrum the eigenvectors have a localized character. In an influential paper \cite{AbrahamsPRL1979}, Anderson and his collaborators gave strong evidence that, in one and two space dimensions, the eigenvectors of a disordered Hamiltonian should {\it all} be localized, unless extraordinary conditions are met. This is the same as saying that two dimensional materials are barred from conducting electricity. But according to the data from Fig.~\ref{Fig:Fig1}, the two dimensional model displays spectral regions where diffusive dynamics still occurs, even when the randomly fluctuating part is five time stronger than the non-fluctuating part!

\vspace{0.2cm}

The resilience against the localization predicted by Anderson is almost always related to topological aspects of the dynamics. In Fig.~\ref{Fig:Fig2}, we report the evaluation of the pairing~\eqref{Eq:SigmaH2} for the model~\eqref{Eq:Model1}, as function of the Fermi energy. The calculation is for disorder strength $W=4$ where the spectral gap does not exist anymore, yet the pairing displays quantized values (the flat non-fluctuating values are quantized with more than 10 digits of precision!). It is clear from this data that something more than the ordinary quantization of the cocycle pairing is at work here. Indeed, the scaling analysis from panels (b) and (c) indicates that the quantization of the pairings holds for all Fermi energies except for the critical values where the de-localized spectral regions were detected in Fig.~\ref{Fig:Fig1}. So it seems that the change in the quantized value of the pairing and the existence of de-localized spectrum are interconnected. The phenomenon was observed in laboratories for real samples, {\it e.g.} in the classic integer quantum Hall experiments \cite{KlitzingPRL1980} but also in samples that are not subjected to magnetic fields \cite{ChangPRL2016,CheckelskyNatPhys2014}, as is the case for the model \eqref{Eq:Model1}.

\begin{figure}
\center
  \includegraphics[width=0.9\textwidth]{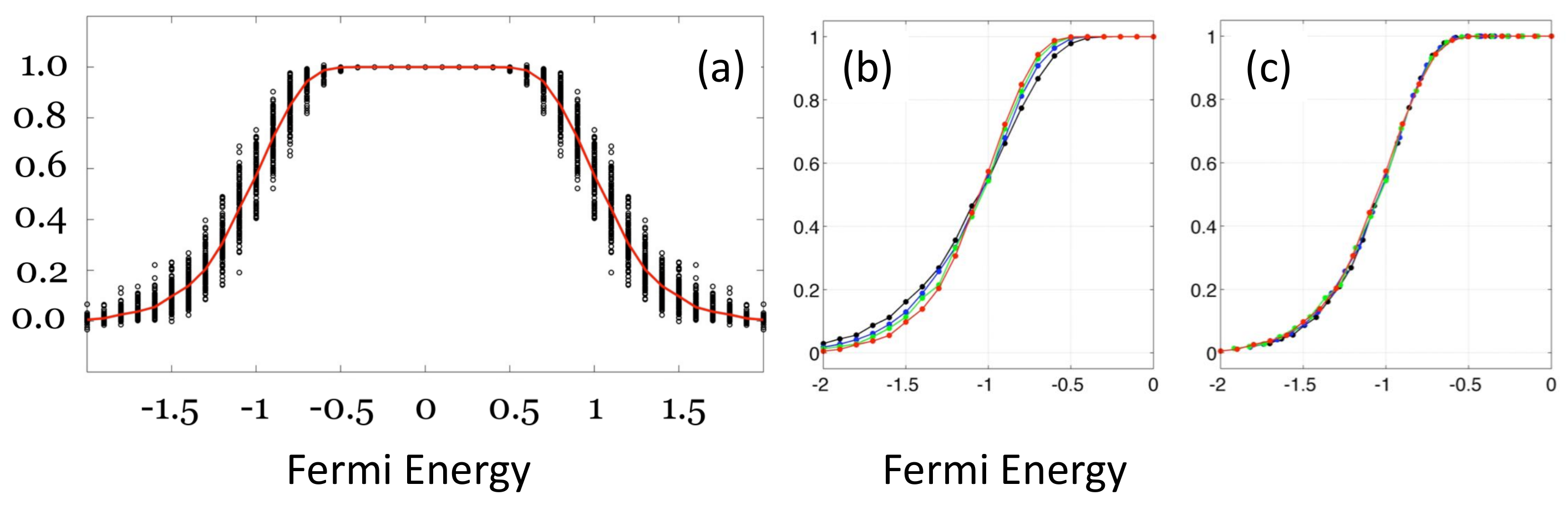}\\
  \caption{\small (Reproduced with permission from \cite{ProdanSpringer2017}) Numerical evaluation of the pairing~\eqref{Eq:SigmaH2} as function of $E_F$, using the algorithm outlined in the text and Hamiltonian~\eqref{Eq:Model1} at disorder strength $W=4$. Panel (a) shows results for individual random values $\xi$ (the dots), as well as the average over 100 disorder configurations (red line). Panel (b) overlays the results obtained for different sample sizes and panel (c) confirms that the curves from panel (b) overlap almost perfectly when the horizontal axis is rescaled, proving the existence of a critical point (see \cite{ProdanSpringer2017} for full details).}
\label{Fig:Fig2}
\end{figure} 

\vspace{0.2cm}

Perhaps a more convincing example is the following one dimensional model generated from the algebra $M_2(\CM) \otimes C\big([-\tfrac{1}{2},\tfrac{1}{2}]^\ZM\big) \rtimes \ZM$:
\begin{equation}
\begin{aligned}
H_{\xi} =  \sum_{x \in \ZM} \big \{\tfrac{1}{2} t(\tau_x \xi) & \big [\sigma_+ \otimes |x\rangle \langle x+1 | \\
& + \sigma_- \otimes |x+1\rangle \langle x | \big ] + m(\tau_x \xi) \sigma_2 \otimes |x\rangle \langle x | \big \},
\end{aligned}
\end{equation}
where $t,m: [-\tfrac{1}{2},\tfrac{1}{2}]^\ZM \to \RM$ are two continuous functions, $\tau$ is the standard shift on their domain and $\sigma$'s are Pauli's matrices. We point out that the model has the key symmetry $\sigma_3 H_\xi \sigma_3 = -H_\xi$.
We are going to solve the Schroedinger equation $H_\xi \psi = E \psi$ at $E=0$,
\begin{equation}
t_x \psi_{x-\alpha}^\alpha + i\alpha m_x \psi_x^\alpha = 0 \ \Rightarrow \  \psi_x^\alpha = \prod_{j=1}^x \left ( t_x/m_x \right ) \psi_0^\alpha,
\end{equation} 
where $\alpha = \pm 1$ is the index for the two components of $\psi$. As observed in \cite{MondragonPRL2014}, we are in a rare situation where the Lyapunov exponent quantifying the asymptotic spatial profile of the solution can be computed explicitly. Indeed
\begin{equation}
\Lambda : = \max_{\alpha = \pm} \left [ - \lim_{x \rightarrow \infty} \tfrac{1}{x} \log |\psi_x^\alpha| \right ]= \Big | \lim_{x \rightarrow \infty} \tfrac{1}{x} \sum_{n=1}^x \Big (\ln |t(\tau_x \xi)| - \ln |m(\tau_x \xi)| \big )\Big |
\end{equation}
and Birkhoff's theorem can be applied at this point. For example, if we take $\PM$ to be the normalized product measure over $[-\tfrac{1}{2},\tfrac{1}{2}]^\ZM$ and $t(\{\xi_x\}) = 1+ W_1\,  \xi_0$ as well as $m(\{\xi_x\}) = m+ W_2 \, \xi_0$, then
\begin{equation}
\Lambda = \left | \ln \left [ \frac{|2+W_1|^{1/W_1+1/2}}{\big |2-W_1|^{1/W_1-1/2}} \frac{ |2m-W_2|^{m/W_2-1/2}}{|2m+W_2|^{m/W_2+1/2}}\right ] \right |.
\end{equation}
The remarkable fact here is the existence of an entire manifold of parameters where $\Lambda$ vanishes. This critical manifold is shown in Fig.~\ref{Fig:Fig3} and, as one can see, it separates the parameter space $(W_1,W_2,m)$ into two disjoint regions.\footnote{It was verified numerically in \cite{MondragonPRL2014} that $\Lambda$ is finite for all parameters except for those belonging to the shown critical manifold.} This critical manifold was recently detected in a laboratory using cold atom experiments \cite{MeierScience2018}, which was one of the highlights in the field for year 2018.

\begin{figure}
\center
  \includegraphics[width=\textwidth]{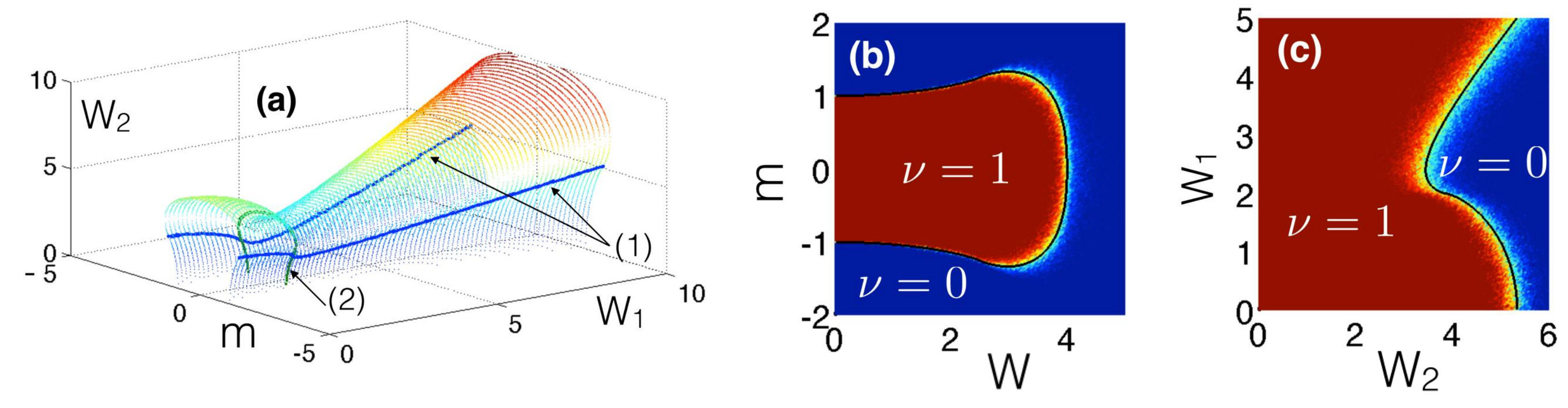}\\
  \caption{\small (Reproduced with permission from \cite{MondragonPRL2014}) (a) The critical surface ($\Lambda = 0$) in the 3-dimensional phase space $(W_1,W_2,m)$. (b,c) Maps of the pairing \eqref{Eq:Niu} for two sections of the phase space, defined by the constraints $W_2=2W_1=W$ (b) and $m=0.5$ (c). The analytic critical manifold is shown in these panels by the black line (see \cite{MondragonPRL2014} for more details).}
\label{Fig:Fig3}
\end{figure} 

\vspace{0.2cm}

Due to its particular symmetry,
\begin{equation}
\frac{H_\xi}{|H_\xi|} = \begin{pmatrix} 0 & U_\xi^\ast \\ U_\xi & 0 \end{pmatrix},
\end{equation}
and the topology of the dynamics is encoded in the unitary element $U_\xi$. When a spectral gap is open around $E=0$,  $U_\xi$ is the left regular representation of a unitary element $u \in C\big([-\tfrac{1}{2},\tfrac{1}{2}]^\ZM\big) \rtimes \ZM$ and we can evaluate the following odd pairing 
\begin{equation}\label{Eq:Niu}
\nu = \imath \, \Tt\big (u \partial u^{-1}\big ).
\end{equation}
As in the previous example, the numerical algorithm is not bound to the spectrally gapped cases and maps of this pairing are shown in Figs.~\ref{Fig:Fig3}(b,c) for entire sections of the parameter space. They reveal again that the pairing remains quantized everywhere except at the critical manifold, where an abrupt change between the quantized values 1 and 0 is observed. Thus, we see again a quantization of the pairing that holds beyond the $C^\ast$-umbrella and a strong connection between the jumps in the quantized values and the emergence of delocalized spectrum.

\vspace{0.2cm}

The connection between the phenomena we just discussed and Noncommutative Geometry was established by Bellissard and his collaborators \cite{BellissardJMP1994}. Let us denote the Fermi projection as $P_\xi$ to emphasize its dependence on the random variable and note that the second moment $\sum_{x \in \ZM^2} |x|^2 |\langle x | P_\xi |0\rangle |^2$ is a measure of the spatial spread of the kernel of this operator. If finite, it gives a length scale for the dynamics at energies near the Fermi level. For a non-disordered system, this quantity diverges as soon as the Fermi level dives into the spectrum, but so it does for a disordered system for $\xi$ in a set of non-zero measure. The quantity that remains finite is the average \cite{AizenmanIM2006}[Th.~1.1]
\begin{equation}\label{Eq:Cond1}
l_F^2 = \int_{\Xi} {\rm d}\PM(\xi) \, \sum_{x \in \ZM^2} |x|^2 |\langle x | P_\xi |0\rangle |^2 < \infty,
\end{equation}
which in the $C^\ast$-language is simply $\Tt(|\nabla p|^2)$. What seems to happen in the simulations we just showed is that the parings of the cocycles remain quantized not only over the smooth sub-algebra of $C(\Xi) \rtimes \ZM^2$, but also over a much larger Sobolev space. As we shall see in subsection~\ref{Sec:Sobolev}, these Sobolev spaces can be regarded as the natural domains of the cyclic cococycles. But this, by no means, guaranties the observed quantization. For the latter, guided by Alaine Connes' quantized calculus \cite{ConnesBook}[Ch.~4], Bellissard and collaborators established in \cite{BellissardJMP1994} and index theorem
\begin{equation}\label{Eq:Ind1}
{\rm Ind} \, F_\xi = 2\pi \imath \sum_{\lambda \in \Ss_2} (-1)^\lambda \Tt \big(p_F(\partial_{\lambda_1}p_F) (\partial_{\lambda_2}p_F)\big )
\end{equation}
with the equality holding $\PM$-almost surely. Above, $F_\xi$ comes from a Dirac Fredholm module and, in a nutshell, condition~\eqref{Eq:Cond1} assures the $(2+\epsilon)$-summability of the Fredholm module, but only $\PM$-almost surely, and the right side of \eqref{Eq:Ind1} is the evaluation of the Connes-Chern character, after the observation that the index is $\PM$-almost surely independent of $\xi$, hence it can be averaged out. The index theorem explains the phenomenon seen in our first example because a change in the quantized value of the pairing can only happen when $l_F \to \infty$.

\vspace{0.2cm}

The evaluation of the Connes-Chern character in \cite{BellissardJMP1994} shunts the Connes-Moscovici local index formula \cite{ConnesGAFA1995} by using a geometric identity due to Alain Connes \cite{ConnesIHESPM1985}. Many people have found inspiration in different parts of this work by Alain that is celebrated in this volume. As strange as it may sound, Lemma~2 from \cite{ConnesIHESPM1985}[Sec.~9], where this geometric identity is proved, had a pivotal influence on my research. Indeed, sometime in 2012, my student Bryan Leung and myself just completed a calculation of the magneto-electric $\alpha_{EM}$ coefficient in three space dimensions and the result showed that, under a cyclic deformation of a model in some parameter space, the variation of $\alpha_{EM}$ per cycle is given by \cite{LeungJPA2013}
\begin{equation}\label{Eq:AlphaEM}
\Delta \alpha_{EM} = \Lambda_4 \sum_{\lambda \in \Ss_4} (-1)^\lambda \Tt \Big(p_F \prod_{j=1}^4 \partial_{\lambda_j} p_F \Big )
\end{equation}
and the right hand side is again a pairing of a cyclic cocycle. At that time, we were not aware of any index theorem that can be related to this pairing, although hints appeared in the literature ({\it e.g.} in a preprint later published in \cite{CareyAMS2014}[Ch.~5]).\footnote{The index theorems, in their explicit form, were established later by Chris Bourne in his thesis \cite{BourneThesis2015}, which also deals with the Clifford indices.} However, in order to push the index theorem into the Sobolev domain, a generalization of Alain's geometric identity was needed. And together with Jean Bellissard and Bryan Leung, we found the following generalization \cite{ProdanJPA2013}: 
\begin{equation}\label{Eq:Id1}
\int_{\RM^d}{\rm d}  w \ \mathrm{tr}_\gamma \Big (\gamma_0 \prod_{i=1}^{d}\gamma \cdot \left(\widehat{y_i-w} -\widehat{y_{i+1}-w} \right )\Big ) 
 =\Lambda_d \sum_{\lambda} (-1)^\lambda \prod_{i=1}^{d} (y_i)_{\lambda_i},
\end{equation}
where $\gamma$'s are the Clifford matrices for even dimension $d$, $y_i \in \RM^d$ and $\hat x = x/|x|$. Many readers will immediately see the connection between the left side of this identity and the Connes-Chern characters derived from a Dirac Fredholm module either over $\RM^d$ or $\ZM^d$, and between the right side and the pairing~\eqref{Eq:AlphaEM}. An equivalent identity for odd dimensions was also derived in \cite{ProdanJFA2016}. The index theorems derived from these identities proved half of the conjectured table of strong topological phases \cite{SchnyderPRB2008,QiPRB2008,Kitaev2009,RyuNJP2010} ({\it i.e.} those labeled by $\ZM$), specifically, it confirmed that all the entries there are thermodynamically distinct because, when passing from one phase to another, one necessarily has to cross an experimentally detectable localization-delocalization phase transition. 

\begin{figure}
\center
  \includegraphics[width=\textwidth]{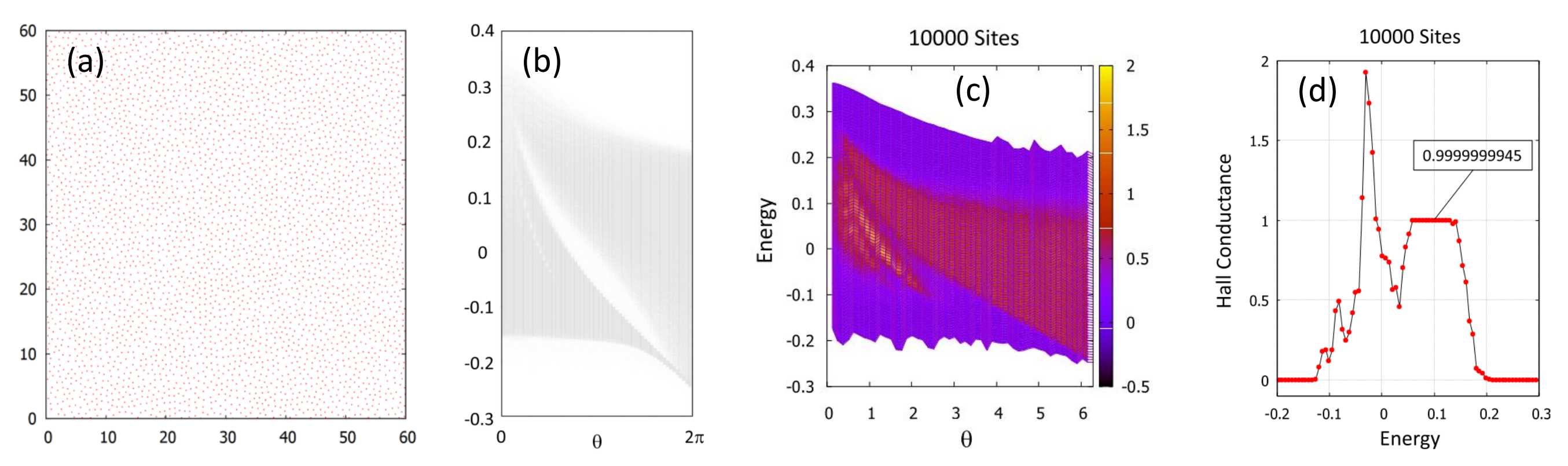}\\
  \caption{\small (Reproduced with permission from \cite{BourneJPA2018}) (a) Section of an amorphous Delone set containing 1000 sites. (b) The spectrum of the Hamiltonian~\eqref{Eq:HAmor} over the pattern from panel (a), rendered as function of the magnetic flux $\theta$. (c) Map of the pairing~\eqref{Eq:SigmaH2} in the plane of the Fermi energy and magnetic flux, as evaluated for the Hamiltonian~\eqref{Eq:HAmor}. (d) A section of the map from panel (c) cut at $\theta=1.5$, confirming a high degree of quantization for the pairing.}
\label{Fig:Fig4}
\end{figure} 

\vspace{0.2cm}

The formalism faced a challenge in 2017, when two research groups observed similar topological phase transitions in amorphous metamaterials \cite{AgarwalaPRL2017,MitchellNatPhys2018}. These models can no longer be generated from a crossed product algebra and the whole program, with both its numerical and theoretical components, had to be upgraded to groupoid algebras (see section~\ref{Sec:GAlg}). Once it was understood that the differential calculus comes from an $\RM^d$-action, the numerical algorithm was implemented almost without modifications and, in Fig.~\ref{Fig:Fig4}, we illustrate the outcome of a simulation \cite{BourneJPA2018}. It was performed for the amorphous pattern $\Ll$ shown in panel (a) subjected to a perpendicular uniform magnetic field. The dynamics was generated by the Hamiltonian
\begin{equation}\label{Eq:HAmor}
H_\Ll : \ell^2(\Ll) \to \ell^2(\Ll), \quad H_\Ll = \sum_{x,x' \in \Ll} e^{\imath \theta x \wedge x'} e^{-3|x -x'|} |x\rangle \langle x'|
\end{equation}
and panel (b) reports the spectrum of this Hamiltonian as function of the magnetic flux $\theta$. The light regions seen in this plot are regions of low spectral density, as opposed to spectral gaps, hence we are again in a gap-less  situation yet the pairing~\eqref{Eq:SigmaH2}, reported in panels (c) and (d), displays again quantized values. The index theorems that are behind this phenomenon were derived in \cite{BourneJPA2018,BourneAHP2020} using the local index formula \cite{ConnesGAFA1995}, as extended to the semifinite noncommutative geometries in \cite{CareyAMS2014}. In the present work, we undertook the challenge of proving the index theorems via a proper generalization of the identity~\eqref{Eq:Id1} (see Lemma~\ref{Lem:Id2}).

\vspace{0.2cm}

We are now coming to the greatest challenge of all, which is  extending the formalism to correlated many-fermion systems. The present paper is geared towards this challenge and reports, in great part, on a contribution that was recently completed by my collaborator Bram Mesland and myself \cite{MeslandArxiv2021}, luckily just in time for this celebratory volume. During this work, we first stepped back and re-assessed the groupoid algebras introduced by Bellissard and Kellendonk in the single fermion picture, \cite{Bellissard1986,Bellissard1995,KellendonkRMP95}, already mentioned above. These algebras were often referred to as the algebras of physical observables, but we took a different point of view and replaced the latter with the algebra $\KM\big (\ell^2(\Ll)\big )$ of compact operators. Then the Bellissard-Kellendonk groupoid algebra canonically associated with the pattern $\Ll$ can be seen as formalizing physically relevant derivations over the algebra $\KM\big (\ell^2(\Ll)\big )$. This apparently minor conceptual shift played an important role because it enabled us to identify more clearly the task when the context is upgraded to that of many-fermion systems: Repeat the program completed by Bellissard and Kellendonk, this time looking for algebras of derivations while replacing $\KM\big (\ell^2(\Ll)\big )$ by ${\rm CAR}\big (\ell^2(\Ll)\big )$, the algebra of canonical anti-commutation relations over the Hilbert space of the lattice. My work with Mesland was restricted to derivations over  ${\rm GICAR}\big (\ell^2(\Ll)\big )$, the sub-algebra of gauge invariant elements, and the reader can find our conclusion reproduced here in subsection~\ref{Sec:MConclusions}. As I will argue in this work, we are now in a position where we can generate interesting and relevant states over this algebra of physical derivations, as well as write down relevant cyclic cocycles and, more importantly, jump-start the numerical simulations of many-body transport coefficients. We are eagerly looking forward to seeing the outcomes of these simulations as well as to working on index theorems for the many-body context.

\vspace{0.2cm}

Unfortunately there was no room or time to survey all the work done in connection with Noncommutative Geometry and materials science, and I apologize to the reader for mostly showcasing works in which I was directly involved. My  primary goal for this somewhat un-conventional introduction was to communicate to a broader audience certain aspects that I believe are really valued by the materials science community. As one can see, I did not focus at all on the classification of states, primarily  because experimental scientists are more interested in learning about observable effects and much less in abstract labels and, secondly, because with meta-materials the possibilities are endless and cannot be captured by a table. Indeed, the algebra of physical observables can be engineered to be any AF-algebra. Lastly, this introduction and the paper itself left out half of the story, specifically, the so called bulk-boundary and bulk-defect correspondences. They refer to extremely robust spectral and dynamical effects that are observed along the interface between two samples with distinct bulk topological invariants (see \cite{ProdanJPA2021} for a unifying reformulation using groupoid algebras). Here again, certain index theorems pushed to Sobolev domains play an important role in establishing the dynamical characteristics  of these interfaces \cite{ProdanSpringer2016}[Sec.~6.6].

\vspace{0.2cm} 

The paper is organized as it follows. Section 2 covers classical concepts related to Delone sets and their many-body covers introduced in \cite{MeslandArxiv2021}. In particular, we review the Bellissard-Kellendonk groupoids canonically associated to a pattern and their natural generalization to the covers we just mentioned. Section 3 investigates the groupoid $C^\ast$-algebras and the bi-equivariant subalgebras relative to a specific 2-action of the symmetric group introduced in \cite{MeslandArxiv2021}. Additionally, this section introduces a differential calculus over the many-body groupoids and writes down cyclic cocycles that are relevant for the transport experiments. Section 4 introduces the algebra of physical derivations over the ${\rm GICAR}$-algebra and summarizes the conclusions from \cite{MeslandArxiv2021}, connecting this algebra and the many-body groupoid algebras introduced in section 3. The paper concludes with a new proof of the index theorems for the pairings of the cyclic cocycles mentioned above for the case $N=1$, based on a generalization of the identity~\eqref{Eq:Id1} that works on generic Delone sets. These can be found in section 5.

\section{Spaces of Patterns and Associated Concepts}
\label{SubSec:Patterns}

This section is about the subsets of the abelian topological group $\RM^d$, their interactions with the group itself and the canonical algebraic and topological structures that emerge from these interactions. Given a Delone point set $\Ll_0$, one considers the topological space $\Omega_{\Ll_0}$ supplied by the closure of its $\RM^d$-orbit in a specific topology. Everything discussed in this section is built on the resulting topological dynamical system $(\Omega_{\Ll_0},\RM^d)$. Specifically, the restriction of its transformation groupoid to a so called canonical transversal supplies the Bellissard-Kellendonk groupoid of the pattern \cite{Bellissard1986,Bellissard1995,KellendonkRMP95}. Then blow-ups of this groupoid, induced by so called many-body covers of $\Omega_{\Ll_0}$, generate canonical many-body groupoids that will be later connected with the dynamics of a fermion gas populating the lattice $\Ll_0$.

\subsection{The metric space of closed $\RM^d$ subsets} The metric space $(\Kk(\RM^d),{\rm d}_{\rm H})$ of compact subsets of $\RM^d$ equipped with the Hausdorff metric is well known and extensively used in pattern and image analysis (see {\it e.g.} \cite{BarnselyBook}). However, when dealing with extended physical systems, as in the present case, one needs to consider the larger space of closed subsets of $\RM^d$. One way to deal with these subsets is to promote them to compact subsets of the one-point compactification $\alpha \RM^d$. We recall that, since $\RM^d$ is proper, $\alpha \RM^d$ is metrizable in an essentially canonical way. Hence, one can put forward the following:

\begin{definition}[\cite{ForrestAMS2002,LenzTheta2003}] The space of patterns consists of the set $\Cc(\RM^d)$ of closed sub-sets of $\RM^d$ endowed with the metric
\begin{equation}\label{Eq:PMetric}
{\rm D}(\Ll,\Ll') =\overline{\rm d}_{\rm H}(\bar \Ll,\bar \Ll'), \quad \Ll,\Ll' \in \Cc(\RM^d),
\end{equation}
where $\bar \Ll$ and $\bar \Ll'$ are the closures of the embeddings of $\Ll$ and $\Ll'$ in the one-point compactification of $\RM^d$ and $\bar{\rm d}_{\rm H}$ is the Hausdorff metric on this compactification.
\end{definition}

The following statements supply the characterization we need for this space:

\begin{proposition}[\cite{LenzTheta2003}] The space of patterns is bounded, compact and complete. Furthermore, there is a continuous action of $\RM^d$ by translations, that is, a homomorphism between topological groups,
\begin{equation}\label{Eq:RDGroupAction}
\mathfrak t : \RM^d \rightarrow {\rm Homeo}\big ( \Cc(\RM^d),D \big ), \quad \mathfrak t_{x}(\Lambda) = \Lambda - x.
\end{equation} 
\end{proposition}

\subsection{Delone sets and and their many-body covers} We will exclusively work with discrete subsets. Specifically:

\begin{definition} Let $\Ll \subset \RM^d$ be discrete and fix $0<r<R \in \RM$. Then:
\begin{enumerate}[\rm \ \ (1)] 
\item $\Ll$ is $r$-uniformly discrete if $|B_r(x) \cap \Ll| \leq 1$ for all $x \in \RM^d$.
\item $\Ll$ is $R$-relatively dense if $|B_R(x) \cap \Ll | \geq 1$ for all $x \in \RM^d$.
\end{enumerate}
An $r$-uniform discrete and $R$-relatively dense set $\Ll$ is called an $(r,R)$-Delone set.
\end{definition}

\begin{remark}{\rm Throughout, if $S$ is a subset of a measure space, then $|S|$ denotes its measure. If $S$ is discrete, then the measure is the counting measure and $|S|$ is the cardinal of $S$. Also, $B_s(x)$ denotes the open ball in $\RM^d$ of radius $s$ and centered at $x$.
}$\Diamond$
\end{remark}

\begin{proposition}[\cite{Bellissard2000}]\label{Pro:rRDelone} The set ${\rm Del}^r_R(\RM^d)$ of $(r,R)$-Delone sets in $\RM^d$ is a compact subset of $(\Cc(\RM^d),D)$.
\end{proposition}

The Delone space  ${\rm Del}^r_R(\RM^d)$ accepts natural covers, which in \cite{MeslandArxiv2021} were dubbed many-body covers for reasons that will become apparent later. We describe them here.

\begin{proposition}[\cite{MeslandArxiv2021}] Fix $r <R$ and consider the set $\widehat {\rm Del}_{(r,R)}^{(N)}(\RM^d)$ of triples $(\chi_V,V,\Ll)$, where $N \in \NM^\times$ and:
\begin{itemize}
\item $\Ll \in {\rm Del}^r_R(\RM^d)$;
\item $V$ is a compact subset of $\Ll$ with $|V|=N$;
\item $\chi_V : \{1,\ldots,N\} \to V$ a bijection.  
\end{itemize}
Then $\widehat {\rm Del}_{(r,R)}^{(N)}(\RM^d)$ can be topologized such that
\begin{equation}
(\chi_V,V,\Ll) \mapsto {\mathfrak a}_N(\chi_V,V,\Ll) : = \Ll 
\end{equation}
becomes a connected (infinite) cover of ${\rm Del}^r_R(\RM^d)$.
\end{proposition}

While we will not discuss the topology of $\widehat {\rm Del}_{(r,R)}^{(N)}(\RM^d)$ explicitly, for orientation, we want to mention that this topology is such that the map $(\chi_V,V,\Ll) \to (V,\Ll)$ is a $N!$ cover if the set of tuples $(V,\Ll)$ is equipped with the topology inherited from $\Kk(\RM^d) \times {\rm Del}_R^r(\RM^d)$. We call this map the order cover while the map ${\mathfrak a}_N$ is referred to as the $N$-body cover, for reasons that will become apparent in the following sections. The hat in $\widehat {\rm Del}_{(r,R)}^{(N)}(\RM^d)$ and other places is there to indicate that we are dealing with a cover of the familiar space ${\rm Del}^r_R(\RM^d)$.

\begin{remark}{\rm The case $N=1$ is special because $\chi_V$ is trivial and can be discarded. In this case, the topology on the tuples $(x,\Ll)$, $x \in \Ll$, is the one inherited from the embedding in $\RM^d \times {\rm Del}_R^r(\RM^d)$.
}$\Diamond$
\end{remark}

The points of $\widehat {\rm Del}_{(r,R)}^{(N)}(\RM^d)$ will be denoted by greek letters like $\xi$, $\zeta$, etc., and the triple corresponding to $\xi$ will take the form $(\chi_\xi,V_\xi,\Ll_\xi)$. The following statement assures us that the $\RM^d$ action extends over the many-body covers:

\begin{proposition}\label{Pro:RDAction1} The group action of $\RM^d$ from Eq.~\eqref{Eq:RDGroupAction} lifts to an action on $\widehat {\rm Del}_{(r,R)}^{(N)}(\RM^d)$ by homeomorphisms
\begin{equation}
\hat {\mathfrak t}_x(\xi) = \big(\mathfrak t_x \circ \chi_\xi,\mathfrak t_x(V_\xi),\mathfrak t_x(\Ll_\xi)\big ).
\end{equation}
\end{proposition}

There is a second group action on the many-body covers, which will play a central role in our framework:

\begin{proposition} Let $\Ss_N$ be the group of permutations of $N$ objects, which we identify here with the group of bijections from $\{1,\ldots,N\}$ to itself. Then $\Ss_N$ acts by homeomorphisms on $\widehat {\rm Del}_{(r,R)}^{(N)}(\RM^d)$ via
\begin{equation}\label{Eq:Ls}
\Lambda_s (\xi) := \big (\chi_\xi \circ s^{-1},V_\xi,\Ll_\xi\big ), \quad s\in \mathcal{S}_N.
\end{equation}
\end{proposition}

\begin{remark}{\rm The action of $\Ss_N$ described above coincide with the group of deck transformations for the order cover $(\chi_V,V,\Ll) \mapsto (V,\Ll)$.
}$\Diamond$
\end{remark}

\subsection{Hulls, transversals and their many-body covers}

\begin{definition}\label{Def:Hull} Let $\Ll_0 \subset \RM^d$ be a fixed Delone point set. The hull of $\Ll_0$ is the topological dynamical system $(\Omega_{\Ll_0},\mathfrak t,\RM^d)$, where
\begin{equation}
\Omega_{\Ll_0}= \overline{\{\mathfrak t_a(\Ll_0)=\Ll_0 - a, \ a\in \RM^d\}} \subset \Cc(\RM^d).
\end{equation}
\end{definition}

The above dynamical system is always transitive but may fail to be minimal. Indeed, $\Omega_\Ll$ may not coincide with $\Omega_{\Ll_0}$ for $\Ll \in \Omega_{\Ll_0}$, but it is always true that $\Omega_\Ll \subseteq \Omega_{\Ll_0}$ for any $\Ll \in \Omega_{\Ll_0}$.

\begin{remark}{\rm Proposition~\ref{Pro:rRDelone} assures us that all the subsets contained in $\Omega_{\Ll_0}$ are $(r,R)$-Delone point sets if $\Ll_0$ is one. Note, however, that the hull can be very well defined for just uniformly discrete patterns, a setting that is relevant for the so called bulk-defect correspondence (see \cite{ProdanJPA2021}).
}$\Diamond$
\end{remark}

As a closed subset of the space of Delone sets, $\Omega_{\Ll_0}$ accepts many-body covers:

\begin{definition} The $N$-body cover above $\Omega_{\Ll_0}$ is defined as
\begin{equation}
\widehat \Omega_{\Ll_0}^{(N)} : = \mathfrak a_N^{-1}(\Omega_{\Ll_0}).
\end{equation}
Proposition~\ref{Pro:RDAction1} assures us that $\big (\widehat \Omega_{\Ll_0}^{(N)}, \hat{\mathfrak t},\RM^d\big )$ is again a topological dynamical system.
\end{definition}

\begin{definition}\label{Def:Transversal} The canonical transversal of a hull $(\Omega_{\Ll_0},\mathfrak t,\RM^d)$ of a Delone set $\Ll_0$ is defined as
\begin{equation*}
\Xi_{\Ll_0} = \{ \Ll \in \Omega_{\Ll_0}, \ 0 \in \Ll\}.
\end{equation*}
The transversal is a closed hence compact subspace of $\Omega_{\Ll_0}$.
\end{definition}

As it was the case for the hull, $\Xi_{\Ll_0}$ accepts many-body covers, which are defined slightly different:

\begin{definition} The $N$-body cover of $\Xi_{\Ll_0}$ is defined as
\begin{equation}\label{Eq:XiCover}
\widehat \Xi_{\Ll_0}^{(N)} : = \big \{\xi \in \widehat \Omega_{\Ll_0}^{(N)}, \ \chi_\xi(1) =0\big \}.
\end{equation} 
\end{definition}

\begin{remark}\label{Re:MBCoverXi}{\rm The $N$-body cover $\widehat \Xi_{\Ll_0}^{(N)}$ is a closed subset of $\widehat {\rm Del}_{(r,R)}^{(N)}(\RM^d)$. Except for the case $N=1$, when the many-body covers reduce to the ordinary hulls and transversals, $\widehat \Omega_{\Ll_0}^{(N)}$ and $\widehat \Xi_{\Ll_0}^{(N)}$ are locally compact but not compact spaces.
}$\Diamond$
\end{remark}

\subsection{Canonical groupoid of a point pattern}

We review here the groupoid introduced by Bellissard and Kellendonk \cite{Bellissard1986,Bellissard1995,KellendonkRMP95}. We start with a few general concepts:

\begin{definition}[Generalized equivalence relations, \cite{WilliamsBook2}~p.~5]\label{Def:GGenEquiv} Let $G$ be a group and $X$ a set. Then $X \times G \times X$ can be given the structure of a groupoid by adopting the inversion function 
\begin{equation}
(x,g,y)^{-1} = (y,g^{-1},x),
\end{equation}
and by declaring that two triples $(x,g,y)$ and $(w,h,z)$ are composible if $y = w$, in which case the product is
\begin{equation}
(x,g,y)\cdot (y,h,z) : = (x,gh,z).
\end{equation} 
If $X$ is a topological space and $G$ is a topological group, then $X \times G \times X$ becomes a topological groupoid if endowed with the product topology.
\end{definition}

\begin{definition}[\cite{WilliamsBook2}~p.6] Let $X$ be a $G$-space. Then the transformation groupoid of $(X,G)$ is the topological sub-groupoid of the generalized equivalence relations introduced in Definition~\ref{Def:GGenEquiv}, consisting of the triples $(x,g,y)$ with the entries obeying the constraint $x=g\cdot y$.
\end{definition}

When applied to the dynamical system $ (\Omega_{\Ll_0},\mathfrak t, \RM^d )$, this definition gives:

\begin{proposition} The transformation groupoid $\widetilde \Gg_{\Ll_0}$ associated to the dynamical system $ (\Omega_{\Ll_0},\mathfrak t, \RM^d)$ consists of triples
\begin{equation}
\big (\mathfrak t_{x}(\Ll),x, \Ll \big ) \in \Omega_{\Ll_0} \times \RM^d \times \Omega_{\Ll_0}
\end{equation}
endowed with the multiplication
\begin{equation}
\big (\mathfrak t_{x' + x}(\Ll),x', \mathfrak t_x(\Ll) \big ) \big (\mathfrak t_x(\Ll),x, \Ll \big )  = (\mathfrak t_{x' +x}(\Ll),x'+x, \Ll \big )
\end{equation}
and inversion
\begin{equation}
\big (\mathfrak t_x(\Ll),x, \Ll \big )^{-1} = \big (\Ll,-x, \mathfrak t_x(\Ll) \big ).
\end{equation}
The topology of the transformation groupoid is that inherited from the product topology of $\Omega_{\Ll_0} \times \RM^d \times \Omega_{\Ll_0}$.
\end{proposition}

The source and the range maps of the groupoid $\widetilde \Gg_{\Ll_0}$ act as
\begin{equation}\label{Eq:SR}
\tilde{\mathfrak s}\big (\mathfrak t_x(\Ll),x, \Ll \big ) = (\Ll,0,\Ll), \quad \tilde{\mathfrak r}\big (\mathfrak t_x(\Ll),x, \Ll \big ) = \big (\mathfrak t_x(\Ll),0,\mathfrak t_x(\Ll)\big ),
\end{equation}
hence its space of units coincides with the hull $\Omega_{\Ll_0}$ of the pattern.

\begin{definition} The topological groupoid $\Gg_{\Ll_0}$ canonically associated to a Delone set $\Ll_0$ is the restriction of $\widetilde \Gg_{\Ll_0}$ to the transversal $\Xi_{\Ll_0} \subset \Omega_{\Ll_0}$:
\begin{equation}\label{Eq:G0}
\Gg_{\Ll_0} : = \tilde {\mathfrak s}^{-1}(\Xi_{\Ll_0}) \cap \tilde{\mathfrak r}^{-1}(\Xi_{\Ll_0}).
\end{equation}
Since $\Xi_{\Ll_0}$ is a closed sub-set, the restriction is indeed a topological groupoid.
\end{definition}

This groupoid can be presented more explicitly as:

\begin{proposition}[\cite{Bellissard1986,Bellissard1995,KellendonkRMP95}]\label{Pro:StandardG1} The topological groupoid canonically associated to a Delone set $\Ll_0$ consists of:
\begin{enumerate}[\ \rm 1.]
\item The set
\begin{equation}
\Gg_{\Ll_0} = \big \{ (x,\Ll), \ \Ll \in \Xi_{\Ll_0}, \ x \in \Ll \big \} \subset \RM^d \times \Xi_{\Ll_0},
\end{equation}
equipped with the inversion map
\begin{equation}
(x,\Ll)^{-1} = \big(-x,\mathfrak t_{x}(\Ll)\big) = \mathfrak t_{x}(0, \Ll).
\end{equation}
\item The subset of $\Gg_{\Ll_0} \times \Gg_{\Ll_0}$ of composable elements
\begin{equation}
 \Big \{\big ((x',\Ll'),(x,\Ll)\big ) \in \Gg_{\Ll_0} \times \Gg_{\Ll_0}, \  \Ll'=\mathfrak t_{x}(\Ll)\Big\}
 \end{equation}
equipped with the composition 
\begin{equation}
\big(x',\mathfrak t_x(\Ll)\big) \cdot (x,\Ll) = (x' +x, \Ll).
\end{equation}
\end{enumerate} 
The topology on $\Gg_{\Ll_0}$ is the relative topology inherited from $\RM^d \times \Xi_{\Ll_0}$.
\end{proposition}

\begin{proof}
From~\eqref{Eq:SR}, one finds
\begin{equation}
\tilde{\mathfrak s}^{-1}(\Ll,0,\Ll) = \big (\mathfrak t_{x}(\Ll),x, \Ll \big ), \quad \tilde{\mathfrak r}^{-1}(\Ll,x,\Ll) = \big (\Ll,-x,\mathfrak t_{x}(\Ll) \big ),
\end{equation}
hence $\big (\mathfrak t_{x}(\Ll),x, \Ll \big )$ belongs to the intersection~\eqref{Eq:G0} if and only if $\Ll$ and $\mathfrak t_x(\Ll)$ both belong to $\Xi_{\Ll_0}$. This automatically constrains $x$ to be on the lattice $\Ll$ and $\Gg_{\Ll_0}$ can then be presented as the sub-groupoid of $\widetilde \Gg_{\Ll_0}$ consisting of the triples
\begin{equation}
\big (\mathfrak t_x(\Ll),x, \Ll \big ) \in \widetilde \Gg_{\Ll_0}, \quad \Ll \in \Xi_{\Ll_0} \subset \Omega_{\Ll_0}, \quad x \in \Ll,
\end{equation}
endowed with the topology inherited from $\widetilde G_{\Ll_0}$, hence from $\Xi_{\Ll_0} \times \RM^d \times \Xi_{\Ll_0}$. Lastly, we can drop the redundant first entry of the triples.
\end{proof}

\begin{remark}{\rm The algebraic structure of $\Gg_{\Ll_0}$ can be also described as coming from the equivalence relation on $\Xi_{\Ll_0}$
\begin{equation*}
R_{\Xi_{\Ll_0}} = \big \{ (\Ll,\Ll') \in \Xi_{\Ll_0} \times \Xi_{\Ll_0}, \ \Ll' = \Ll -a \ \mbox{for some} \ a \in \RM^d \big \}.
\end{equation*}
However, the topology of $\Gg_{\Ll_0}$ is not the one inherited from $\Xi_{\Ll_0} \times \Xi_{\Ll_0}$.
}$\Diamond$
\end{remark}

\begin{remark}{\rm Since $\Xi_{\Ll_0}$ is an abstract transversal for $(\Omega_{\Ll_0},\RM^d)$, the transformation groupoid $\widetilde \Gg_{\Ll_0}$ and $\Gg_{\Ll_0}$ are equivalent \cite{MuhlyJOT1987}. In particular, their corresponding $C^\ast$-algebras are Morita equivalent.
}$\Diamond$
\end{remark}

\begin{remark}{\rm In a typical situation, the quantum resonators carry internal structures and, to account for them, the hull and the transversal are multiplied by a finite set. This complication can be avoided by duplicating the pattern such that the transversal of the new pattern becomes $\Xi_{\Ll_0}\times \{1,2,\ldots\}$. This will be tacitly assumed in the following.
}$\Diamond$
\end{remark}

\subsection{Invariant measures} The topological dynamical system $(\Omega_{\Ll_0},\mathfrak t,\RM^d)$, like any other topological dynamical system, accepts invariant and ergodic measures. The measure with physical meaning, denoted here by ${\rm d}\bar \PM$, is fixed by $\Ll_0$ itself:
\begin{equation}
\int_{\Omega_{\Ll_0}} {\rm d}\bar \PM(\Ll) \, \phi(\Ll) : = \lim_{V \to \RM^d}\frac{1}{|V|} \int_{V \subset \RM^d} {\rm d} y \, \phi\big ( \mathfrak t_y(\Ll_0)\big ).
\end{equation}

The real interest, however, is in invariant and ergodic measures on the transversal $\Xi_{\Ll_0}$ relative to the groupoid action. We describe them briefly here and we start from the following reassuring statement:

\begin{proposition}[\cite{ConnesSpringer1979,BellissardCMP2006}] There is a one-to-one correspondence between measures on $\Omega_{\Ll_0}$ invariant under the $\RM^d$-action and measures on the transversal $\Xi_{\Ll_0}$ invariant under the groupoid action.
\end{proposition}

The transition between the two spaces is achieved as follows. One considers first the sub-set $\Xi_{\Ll_0}(s)$ of $\Omega_{\Ll_0}$ consisting of those $\Ll$'s with the property that there is a closed ball $\bar B_x(s)$ of radius $s$ centered at one of the points $x$ of $\Ll$ that contains the origin of $\RM^d$. Since all $\Ll \in \Omega_{\Ll_0}$  belong to ${\rm Del}_R^r(\RM^d)$, if we take $s <r/2$, then the point sets included in  $\Xi_{\Ll_0}(s)$ have exactly one point such that $0 \in \bar B_s(x)$. Then we can define the following canonical map
\begin{equation}
\mathfrak e : \Xi_{\Ll_0}(s) \to \Xi_{\Ll_0}, \quad \mathfrak e(\Ll) = \mathfrak t_x(\Ll),
\end{equation}
where $x \in \Ll$ is the unique point mentioned above and it is easy to check that $\mathfrak e$ respects the groupoid action
\begin{equation}
\mathfrak t_y\big (\mathfrak e(\Ll)\big ) = \mathfrak e \big ( \mathfrak t_y(\Ll) \big ), \quad y \in \Ll.
\end{equation}
Then the trace of $\bar \PM$ on $\Xi_{\Ll_0}(s)$ can be pushed forward through $\mathfrak e$ to obtain the corresponding groupoid invariant measure $\PM : = \mathfrak e_\ast \bar \PM$ on $\Xi_{\Ll_0}$. If $\bar \PM$ was ergodic in the first place, then, for any measurable $\phi : \Xi_{\Ll_0} \to \CM$,
\begin{equation}\label{Eq:Measures}
\begin{aligned}
\int_{\Xi_{\Ll_0}}{\rm d}\PM(\Ll) \, \phi(\Ll) & : = \int_{\Xi_{\Ll_0}(s)}{\rm d}\bar \PM(\Ll) \, (\phi\circ \mathfrak e)(\Ll) \\
& = \lim_{V \to \RM^d}\frac{1}{|V|} \int_V {\rm d} y \, (\phi\circ \mathfrak e)\big ( \mathfrak t_y(\Ll)\big ) \\
& = \lim_{V \to \RM^d}\frac{|\bar B_s(0)|}{|V|} \sum_{x \in \Ll \cap V} \phi(\mathfrak t_x \Ll),
\end{aligned}
\end{equation}
with the equality holding $\PM$-almost surely. We normalize $\PM$ such that $\PM(\Xi_{\Ll_0})=1$ by dividing with the factor $|\bar B_s(0)| \, {\rm Dens}_{\Ll}$, where 
\begin{equation}
{\rm Dens}_{\Ll} = \lim_{V \to \RM^d}\frac{1}{|V|} \sum_{x \in \Ll \cap V} 1
\end{equation} 
is the macroscopic density of the point set $\Ll$. This quantity is non-fluctuating, {\rm i.e.} it coincides with ${\rm Dens}_{\Ll_0}$, $\PM$-almost surely. In fact, we will fix the units such that ${\rm Dens}_{\Ll_0}=1$. With this choice,
the Borel measure on $\RM^d$ defined by
\begin{equation}
|V| : = \int_{\Omega_{\Ll_0}} {\rm d}\bar \PM(\Ll) \, |\Ll \cap V|,
\end{equation}
coincides with the Haar measure of $\RM^d$. This detail will become relevant for the proof of Lemma~\ref{Lem:Id2} containing an important geometric identity.

\subsection{Canonical many-body groupoids}
\label{Sec:CanonicalG}

The groupoid canonically associated to a Delone pattern, as presented in Proposition~\ref{Pro:StandardG1}, accepts the following generalizations over the many-body covers:

\begin{definition}[\cite{MeslandArxiv2021}] The groupoid associated to the dynamics of $N$ fermions over a fixed Delone set $\Ll_0$ consists of:
\begin{enumerate}[\quad \rm 1.] 
\item The topological space 
\begin{equation}
\label{eq: setGn}
\widehat \Gg_{\Ll_0}^{(N)} = \big \{(\zeta,\xi) \in \widehat \Omega_{\Ll_0}^{(N)}\times \widehat \Omega_{\Ll_0}^{(N)}, \ \Ll_\xi = \Ll_\zeta \in \Xi_{\Ll_0}, \  \chi_\xi(1) = 0 \big \},
\end{equation}
equipped with the topology inherited from $\widehat {\rm Del}_{(r,R)}^{(n)}(\RM^d) ^{\times 2}$ and the inversion map
\begin{equation}
(\zeta,\xi)^{-1} = \hat{\mathfrak{t}}_{\chi_\zeta(1)} (\xi,\zeta).
\end{equation}

\item The set of composable elements 
\begin{equation}
\Big \{\big (\hat{\mathfrak t}_{\chi_\zeta(1)}(\zeta',\zeta),(\zeta,\xi)\big ) \in \widehat \Gg_{\Ll_0}^{(N)} \times \widehat \Gg_{\Ll_0}^{(N)} \Big \},
\end{equation}
equipped with the composition map
\begin{equation}
\hat{\mathfrak t}_{\chi_\zeta(1)}(\zeta',\zeta) \cdot (\zeta,\xi) := (\zeta',\xi).
\end{equation}
\end{enumerate}
\end{definition}

\begin{remark}{\rm It is obvious that, for $N=1$, the above groupoid reduces to the standard groupoid $\Gg_{\Ll_0}$ of the pattern.
}$\Diamond$
\end{remark}

\begin{proposition}[\cite{MeslandArxiv2021}] $\widehat \Gg_{\Ll_0}^{(N)}$ is a second-countable, Hausdorff \'etale groupoid.
\end{proposition}

\begin{remark}{\rm Of course, the above statement is good news because an \'etale groupoid comes with a standard system of Haar measures, hence with a standard groupoid $C^\ast$-algebra, which is necessarily separable \cite{SimsSzaboWilliamsBook2020}.  As such, the $K$-theories of these algebras are countable, a fact that gives us hopes about classifying the topological dynamical phases of $N$ fermions hopping over a given lattice.
}$\Diamond$
\end{remark}

There is an alternative characterization of the many-body groupoids, based on the following construction:

\begin{definition}[\cite{WilliamsBook2}] Let $\Gg$ be a locally compact Hausdorff groupoid with open range map. Suppose that $Z$ is locally compact Hausdorff and that $f: Z \to \Gg^{(0)}$ is  a continuous open map. Then
\begin{equation}
\Gg[Z] : =\big \{ (z,\gamma,w)\in Z \times \Gg \times Z: 
 f(z)=r(\gamma) \ {\rm and} \ s(\gamma) = f(w) \big \}
 \end{equation}
 is a topological groupoid when considered with the natural operations 
\begin{equation}
 (z,\gamma,w) (w,\eta,x)=(z,\gamma \eta,x) \ {\rm and} \ (z,\gamma,w)^{-1}= (w,\gamma^{-1},z),
\end{equation}
 and the topology inherited from $Z \times \Gg \times Z$.
 \end{definition}
 
\begin{remark}{\rm The space of units for $\Gg[Z]$ is
 $$
 \Gg[Z]^{(0)} = \big \{ (z,f(z),z), \ z \in Z\big \},
 $$
 hence it can be naturally identified with $Z$. For this reason, the map $f$ is called a blow-up of the unit space and the groupoid $\Gg[Z]$ is referred to as the blow-up of $\Gg$ through $f$.
 }$\Diamond$
 \end{remark}
 
Now, if $\hat{\mathfrak s}$ and $\hat{\mathfrak r}$ denote the source and the range maps of the groupoid $\widehat \Gg_{\Ll_0}^{(N)}$, respectively, then
\begin{equation}\label{Eq:SRM}
\hat{\mathfrak s}\big (\zeta,\xi\big ) = (\xi,\xi), \quad 
\hat{\mathfrak r}\big (\zeta,\xi\big ) = \hat{\mathfrak{t}}_{\chi_\zeta(1)} (\zeta,\zeta).
\end{equation}
As such, the space of units of the groupoid $\widehat \Gg_{\Ll_0}^{(N)}$ coincides with the $N$-body cover $\widehat \Xi_{\Ll_0}^{(N)}$ of the transversal, introduced in Eq.~\eqref{Eq:XiCover}.

 \begin{proposition}[\cite{MeslandArxiv2021}]\label{Pro:GMor} Consider the canonical groupoid $\Gg_{\Ll_0}$. Then  
 \begin{equation}\label{Eq:BlowMap}
 \widehat \Xi_{\Ll_0}^{N} \ni \xi \mapsto \mathfrak u(\xi): = \Ll_\xi \in \Xi_{\Ll_0}
 \end{equation}
 is a blow-up map for $\Gg_{\Ll_0}$ and the associated blow-up groupoid can be presented as
\begin{equation}
\begin{aligned}
\widetilde \Gg_{\Ll_0}^{(N)} = \big \{ \big (\zeta, (x,\Ll), \xi\big),&  \ (x,\Ll) \in \Gg_{\Ll_0}, \  \xi \in \mathfrak a_N^{-1}(\Ll), \\
&  \ \zeta \in \mathfrak a_N^{-1}\big(\mathfrak t_x(\Ll)\big), \ \chi_\xi(1) = \chi_\zeta(1)=0 \big \}.
\end{aligned}
\end{equation}
Furthermore, the map
\begin{equation}\label{Eq:LastIso}
\widehat \Gg_{\Ll_0}^{(N)} \ni (\zeta,\xi) \mapsto \big (\mathfrak t_{\chi_\zeta(1)}(\zeta), \big (\chi_\zeta(1),\Ll_\xi \big ),\xi  \big )\in \widetilde \Gg_{\Ll_0}^{(N)}
\end{equation}
is an isomorphism of topological groupoids.
\end{proposition}

\section{Canonical Groupoid $C^\ast$-Algebras and Associated Structures} 
\label{Sec:GAlg}

The groupoids introduced in the previous section supply canonical $C^\ast$-algebras \cite{RenaultBook}, which will be briefly reviewed here. The focus of the section, however, is on natural 2-actions of the permutation groups and identification of certain bi-equivariant sub-algebras that are directly related to the dynamics of many fermions. As we shall see, these sub-algebras can be endowed with differential structures induced by natural $\RM^d$-actions and canonical semi-finite traces that are invariant against these actions. This enables us to define standard cyclic cocycles that are connected with the physical transport coefficients of many-fermion systems.

\subsection{Groupoid $C^\ast$-algebras and their regular representations} It will be useful to spell out the algebraic operations of the reduced groupoid algebras. For this, we remind that, on the core $C_c\big (\widehat \Gg_{\Ll_0}^{(N)}\big )$ containing the compactly supported $\CM$-valued continuous functions on $\widehat \Gg_{\Ll_0}^{(N)}$, the associative multiplication works as
\begin{equation}
(f_1 \ast f_2)(\zeta,\xi) : =\sum_{(\zeta',\xi')\in \hat{\mathfrak r}^{-1}\big(\hat{\mathfrak r}(\zeta,\xi)\big )} f_{1}(\zeta',\xi') \cdot f_{2}\big ((\zeta',\xi')^{-1}(\zeta,\xi)\big ).
\end{equation}
From~\eqref{Eq:SRM}, one finds
\begin{equation}
\hat{\mathfrak r}^{-1}\big(\hat{\mathfrak r}(\zeta,\xi)\big ) = \big \{ \hat{\mathfrak t}_{\chi_\eta(1)}(\zeta,\eta), \ \eta \in \mathfrak a_N^{-1}(\Ll_\xi) \big \}
\end{equation}
and, since $\big (\hat{\mathfrak t}_{\chi_\eta(1)}(\zeta,\eta)\big )^{-1} =\hat{\mathfrak t}_{\chi_\zeta(1)}(\eta,\zeta)$, the multiplication can be written more explicitly as
\begin{equation}\label{Eq:Conv}
(f_1 \ast f_2)(\zeta,\xi) =\sum_{\eta \in \mathfrak a_N^{-1}(\Ll_\xi)} f_{1}\big(\hat{\mathfrak t}_{\chi_\eta(1)}(\zeta,\eta)\big)  \cdot f_{2} (\eta, \xi ).
\end{equation}
The $\ast$-operation is
\begin{equation}\label{Eq:Inv}
f^\ast(\zeta,\xi) = f\big((\zeta,\xi)^{-1}\big )^\ast=f\big (\hat{\mathfrak t}_{\chi_\zeta(1)}(\xi,\zeta) \big)^\ast.
\end{equation}

Now, if $f \in C_c\big (\widehat \Gg_{\Ll_0}^{(N)}\big)$ and $\varphi \in C_0\big (\widehat\Xi_{\Ll_0}^{(N)}\big)$, the $C^\ast$-algebra of $\CM$-valued continuous  and compactly supported functions over the space of units, then the product
\begin{equation}
(f\, \varphi)(\zeta,\xi) = f(\zeta,\xi) \cdot \varphi\big (\hat{\mathfrak s}(\zeta,\xi)\big )=f(\zeta,\xi) \cdot \varphi (\xi).
\end{equation}
makes $C_c\big (\widehat\Gg_{\Ll_0}^{(N)}\big)$ into a right module over $C_0\big (\widehat \Xi_{\Ll_0}^{(N)}\big)$ and the restriction map 
\begin{equation}
\rho :C_c\big (\widehat \Gg_{\Ll_0}^{(N)}\big ) \rightarrow C_0 \big(\widehat\Xi_{\Ll_0}^{(N)}\big ), \ \ (\rho f)(\xi) = f(\xi,\xi),
\end{equation}
supplies an expectation. By composing with the evaluation maps
\begin{equation}
j_{\eta}: C\big (\widehat \Xi_{\Ll_0}^{(N)}\big ) \to \CM, \ \ j_\eta(\phi) = \phi(\eta), \quad \eta \in \widehat \Xi_{\Ll_0}^{(N)},
\end{equation}
one obtains a family of states $\rho_\eta = j_\eta \circ \rho$ on $C^\ast_r\big (\widehat \Gg_{\Ll_0}^{(N)}\big) $, indexed by the transversal $\widehat \Xi_{\Ll_0}^{(N)}$. The GNS-representations corresponding to these states supply the left-regular representations of $C^\ast_r\big (\widehat \Gg_{\Ll_0}^{(N)}\big)$. They are carried by the Hilbert spaces 
\begin{equation}
\Hh_\eta : =  \ell^2\big (\hat{\mathfrak s}^{-1}(\eta)\big ) = \ell^2\big(\mathfrak a_N^{-1}(\Ll_\eta) \big ),
\end{equation}
and take the explicit form
\begin{equation}\label{Eq:PiL}
[\pi_\eta(f) \psi](\zeta) = \sum_{\xi \in \mathfrak a_N^{-1}(\Ll_\eta)}f\big(\hat{\mathfrak t}_{\chi_\xi(1)}(\zeta,\xi)\big)  \cdot \psi (\xi).
\end{equation}
We recall that these left regular representations supply the $C^\ast$-norm on $C^\ast_r\big (\widehat \Gg_{\Ll_0}^{(N)}\big)$, 
\begin{equation}
\|f\| = \sup_{\eta \in \widehat \Xi_{\Ll_0}^{(N)}} \| \pi_\eta(f)\|.
\end{equation}
These representations will be later compared with the Fock representation of the physical Hamiltonians (see section~\ref{Sec:PhysHam}).

\subsection{The 2-actions and the bi-equivariant groupoid algebras}
\label{Sec:2Action}

There is a direct consequence of Proposition~\ref{Pro:GMor}:

\begin{corollary}[\cite{WilliamsBook2},~Th.~2.52] The groupoid $C^{*}$-algebras $C^\ast_r(\Gg_{\Ll_0})$ and $C^\ast_r\big(\widehat \Gg_{\Ll_0}^{(N)}\big)$ are Morita equivalent and, as a consequence, their $K$-theories coincide.
\end{corollary}

The interest is, however, not entirely on $C^\ast_r\big(\widehat\Gg_{\Ll_0}^{(N)}\big)$, but rather on a bi-equivariant sub-algebra discussed next. We first recall that every topological groupoid $\Gg$ comes with a topological group of bisections $\mathcal{S}(\mathcal{G})$, consisting of continuous maps
\begin{equation*}
\mathcal{S}(\mathcal{G}):=\left\{b:\mathcal{G}^{(0)}\to \mathcal{G}: \mathfrak s\circ b=\textnormal{Id},\quad \mathfrak r\circ b \textnormal{ is a homeomorphism}\right\}
\end{equation*}
and a standard group structure. Above, $\Gg^{(0)}$ is the space of units, and $\mathfrak s$, $\mathfrak r$ are the source and range maps of the groupoid, respectively.

\begin{definition}[\cite{MeslandArxiv2021}]\label{Def:2ActionGroupoid} 
Let $\Gamma$ be a discrete group and $\mathcal{G}$ a locally compact Hausdorff groupoid. A 2-action of $\Gamma$ on $\mathcal{G}$ is a group homomorphism $\alpha:\Gamma\to \mathcal{S}(\mathcal{G})$.
\end{definition}  

\begin{proposition}[\cite{MeslandArxiv2021}] A 2-action of $\Gamma$ on $\mathcal{G}$ automatically induces commuting left and right actions of $\Gamma$ on $\Gg$ via:
\begin{equation}\label{Eq:2Action1}
\gamma_{1}\cdot g \cdot\gamma_{2}:= \alpha_{\gamma_{1}}\big(\mathfrak r(g)\big)\cdot g \cdot\alpha_{\gamma_{2}^{-1}}\big(\mathfrak s(g)\big)^{-1}.
\end{equation}
\end{proposition}

\begin{proposition}[\cite{MeslandArxiv2021}]\label{Pro:2Action} Let $s \in\mathcal{S}_{N}$ be a permutation. The formula
\begin{equation}
\label{eq: 2actiongroupoid}
\alpha_s(\xi):=\big(\Lambda_s(\xi),\xi\big)
\end{equation}
defines a homomorphism $\alpha:\mathcal S_{N}\to \mathcal{S}(\mathcal{G}_{N})$ and thus a 2-action of $\mathcal{S}_{N}$ on $\widehat \Gg_{\Ll_0}^{(N)}$. The induced left and right actions, denoted by $s_1 \cdot (\xi,\zeta) \cdot s_2$, for $s_i \in \mathcal{S}_N$,  are given by
\begin{equation}\label{Eq:2Action2}
s_1 \cdot (\zeta,\xi) \cdot s_2  = \hat{\mathfrak{t}}_{\chi_{\xi}\circ s_2(1)}\big (\Lambda_{s_1}(\zeta),\Lambda_{s_2}^{-1}(\xi)\big ).
\end{equation}
\end{proposition}

For a $C^{*}$-algebra $\Aa$, we denote by $UM(\Aa)$ the unitary group of the multiplier algebra of $\Aa$. There is a similar notion of 2-action of a group on $C^{*}$-algebras:

\begin{definition}[\cite{MeslandArxiv2021}] \label{2actionalgebra}
Let $\Gamma$ be a discrete group and $\Aa$ a $C^{*}$-algebra. A 2-action of $\Gamma$ on $\Aa$ is a group homomorphism $u:\Gamma \to UM(\Aa)$ from $\Gamma$ to the unitary multipliers on $\Aa$. 
\end{definition}

\begin{remark}{\rm A 2-action of $\Gamma$ on $\Aa$ induces left and right actions on $\Aa$ by setting
\begin{equation}
\gamma_{1}\cdot a \cdot \gamma_{2}:=u_{\gamma_{1}}au_{\gamma_{2}}^\ast,
\end{equation}
for $\gamma_i \in \Gamma$ and $a \in \Aa$.
}$\Diamond$
\end{remark}

\begin{lemma}[\cite{MeslandArxiv2021}] Given 2-actions of $\Gamma$ on a $C^{*}$-algebra $\Aa$ and groupoid $\mathcal{G}$,
 we write
 \begin{equation} 
 C_{c,\Gamma}(\mathcal{G},\Aa):=\Big\{f\in C_{c}(\mathcal{G},\Aa): f(\gamma_{1}\cdot g \cdot\gamma_{2})=\gamma_{1}\cdot f(g)\cdot\gamma_{2}\Big\} 
 \end{equation}
 for the space of $\Gamma$-bi-equivariant and compactly supported maps from $\mathcal{G}$ to $\Aa$. Then the subspace $C_{c,\Gamma}(\mathcal{G},\Aa)$ is a $*$-subalgebra of $C_{c}(\mathcal{G},\Aa)$.
\end{lemma}

\begin{definition}[\cite{MeslandArxiv2021}] Let $\Gamma$ be a discrete group. Suppose that $\mathcal{G}$ is a locally compact Hausdorff \'etale groupoid and $\Aa$ is a $C^{*}$-algebra, both of which carry a 2-action by $\Gamma$.
The $\Gamma$-bi-equivariant reduced $\Aa$-$C^{*}$-algebra $C^{*}_{r,\Gamma}(\mathcal{G},\Aa)$ of $\mathcal{G}$ is the closure of the image of $C_{c,\Gamma}(\mathcal{G},\Aa)$ inside $C^{*}_{r}(\mathcal{G},\Aa)$.
\end{definition}

Applied to our context, this gives:

\begin{definition}[\cite{MeslandArxiv2021}] The bi-equivariant groupoid $C^\ast$-algebra associated with the dynamics of $N$ fermions hopping on the Delone sets from the transversal of a $\Ll_0$ is the $\Ss_N$-bi-equivariant sub-algebra $C^{*}_{r,\mathcal{S}_N} (\widehat\Gg_{\Ll_0}^{(N)}, \CM)$, where the 2-action of $\mathcal{S}_N$ on $\CM$ is simply $\Ss_N \ni s \mapsto (-1)^s \in UM(\CM)$.
\end{definition}

The left regular representations of the bi-equivariant sub-algebra drop to representations on the Fock spaces:

\begin{proposition}\label{Pro:FockDescent} Let $|\xi \rangle \! \rangle$ denote the canonical basis of $\ell^2\big(\mathfrak a_N^{-1}(\Ll_\eta) \big )$. Then, since
\begin{equation}
\langle \! \langle \Lambda_{s_1}(\zeta) | \pi_\eta(f) | \Lambda_{s_2}^{-1}(\xi) \rangle \! \rangle = f\big(s_1 \cdot \mathfrak t_{\chi_\xi(1)}(\zeta,\xi) \cdot s_2\big ),
\end{equation}
the sub-space $\Nn$ spanned by 
\begin{equation}
|\xi\rangle \! \rangle -(-1)^s|\Lambda_s(\xi)\rangle \! \rangle, \quad \xi \in \mathfrak a_N^{-1}(\Ll_\xi), \quad s\in \Ss_N,
\end{equation} 
is contained in the kernel of the left regular representation $\pi_\eta$ of $C^{*}_{r,\mathcal{S}_N} \big (\widehat\Gg_{\Ll_0}^{(N)}, \CM\big)$. As such, $\pi_\eta$ drops to a representation on the quotient space $\ell^2\big(\mathfrak a_N^{-1}(\Ll_\eta) \big )/\Nn$, which coincides with the $N$-fermion sector of the Fock space over the lattice $\Ll_\eta$ (see section~\ref{Sec:PhysHam}).
\end{proposition}

\subsection{$\RM^d$ actions, invariant semi-finite traces and cyclic cocycles} We will step back to the groupoid $C^\ast$-algebra $C^\ast_r\big(\widehat \Gg_{\Ll_0}^{(N)}\big)$ and introduce first a differential structure:

\begin{proposition}\label{Pro:NRD} $C^\ast_r\big(\widehat \Gg_{\Ll_0}^{(N)}\big)$ accepts a natural $\RM^d$-action
\begin{equation}\label{Eq:RdAction}
[k\cdot f](\zeta,\xi) = e^{\imath N \langle k, x_\zeta - x_\xi\rangle} f(\zeta,\xi), \quad k \in \RM^d,
\end{equation}
where $x_\eta$ is the center of mass of the subset $V_\eta$ for any $\eta \in \widehat{\rm Del}_{(r,R)}^{(N)}(\RM^d)$ and $\langle , \rangle$ denotes the standard scalar product of $\RM^d$. 
\end{proposition}
\begin{proof} We need to show that the rule from Eq.~\eqref{Eq:RdAction} defines automorphisms. From \eqref{Eq:Conv}, we have
\begin{equation}
[k\cdot (f_1 \ast f_2)](\zeta,\xi) =e^{\imath N \langle k, x_\zeta - x_\xi\rangle}\sum_{\eta \in \mathfrak a_N^{-1}(\Ll_\xi)} f_{1}\big(\hat{\mathfrak t}_{\chi_\eta(1)}(\zeta,\eta)\big)  \cdot f_{2} (\eta, \xi ).
\end{equation}
Then the statement follows from the observation that
\begin{equation}
x_\zeta - x_\xi = x_{\hat{\mathfrak t}_{y}(\zeta)} - x_{\hat{\mathfrak t}_{y}(\eta)} + x_\eta - x_\xi
\end{equation}
holds for any $y \in \RM^d$.
\end{proof}

\begin{remark}{\rm As we shall see in the next section, this action is related to familiar $U(1)$-gauge transformations on the ${\rm CAR}$ algebra over $\ell^2(\Ll_0)$. As such, they are relevant to the computation of the physical transport coefficients.
}$\Diamond$
\end{remark}

\begin{definition}The $\RM^d$-action supplies a set of $d$ derivations,
\begin{equation}
(\partial_j f)(\zeta,\xi) = \imath N (x_\zeta - x_\xi)_j f(\zeta,\xi), \quad j=1,\ldots,d,
\end{equation} 
and we will denote by $C^\infty\big (\widehat \Gg_{\Ll_0}^{(N)} \big )$ the Fr\'echet subalgebra of the smooth elements from $C^\ast_r\big ( \widehat \Gg_{\Ll_0}^{(N)}\big )$.
\end{definition}

\begin{remark}{\rm Since the $\RM^d$-action $\tau$ commutes with the 2-action of the permutation group, these derivations drop on the sub-algebra $C^{*}_{r,\mathcal{S}_N} \big(\widehat \Gg_{\Ll_0}^{(N)}, \CM\big)$.
}$\Diamond$
\end{remark}

We now turn our attention to a particular semi-finite trace. First, we endow the (discrete) fibers of the blow-up map~\eqref{Eq:BlowMap} with the counting measure. The composition with a $\Gg_{\Ll_0}$-invariant measure on $\Xi_{\Ll_0}$ then supplies a $\widehat \Gg_{\Ll_0}^{(N)}$-invariant measure on  $\widehat \Xi_{\Ll_0}^{(N)}$.  

\begin{proposition} Let $\PM$ be a $\Gg_{\Ll_0}$-invariant measure on $\Xi_{\Ll_0}$. Then 
\begin{equation}\label{Eq:HatTrace}
\widehat\Tt_{\Ll_0}^{(N)}(f) : = \int_{\Xi_{\Ll_0}} {\rm d}\PM(\Ll)  \sum_{\xi \in \mathfrak u^{-1}(\Ll)} f(\xi,\xi)
\end{equation}
is a positive linear map over the core algebra $C_c\big(\widehat \Gg_{\Ll_0}^{(N)}\big)$, whose kernel includes the commutator sub-space. As such, $\widehat\Tt_{\Ll_0}^{(N)}$ can be promoted to a semi-finite trace on the $C^\ast$-algebra of the $N$-body groupoid. 
\end{proposition}

\begin{proposition} If $\PM$ is ergodic and $f$ is in the domain of the trace, then $\PM$-almost surely,
\begin{equation}\label{Eq:TrV}
\widehat\Tt_{\Ll_0}^{(N)}(f) = \lim_{V \to \RM^d}\frac{1}{|V \cap \Ll|} \sum_{\xi \in \mathfrak a_N^{-1}(\Ll)}^{\chi_\xi(1) \in V} \langle \xi |\pi_\eta(f)|\xi\rangle,
\end{equation}
for any $\eta \in \mathfrak u^{-1}(\Ll)$.
\end{proposition}
\begin{proof} We have from \eqref{Eq:Measures}, that $\PM$-almost surely,
\begin{equation}
\int_{\Xi_{\Ll_0}} {\rm d} \PM(\Ll) \, \phi(\Ll) = \lim_{V \to \RM^d}\frac{1}{|V \cap \Ll|} \sum_{x\in V\cap \Ll} \phi(\mathfrak t_x \Ll).
\end{equation}
Then the statement follows from definition~\eqref{Eq:HatTrace} of the trace and the expression~\eqref{Eq:PiL} of the left regular representations.
\end{proof}

\begin{remark}{\rm As we shall see, $\widehat\Tt_{\Ll_0}^{(N)}$ essentially matches the trace per volume defined for the Fock representation of the physical Hamiltonians.
}$\Diamond$
\end{remark}

Since the semi-finite trace is invariant against the $\RM^d$-action supplying the derivations, the following expressions supply cyclic cocycles over $C^\infty\big(\widehat \Gg_{\Ll_0}^{(N)}\big)$:
\begin{equation}\label{Eq:CCoC}
\Sigma_J(f_0,f_1,\ldots,f_{|J|}) : = \Lambda_J \sum_{\lambda \in \Ss_J} (-1)^\lambda  \ \widehat\Tt_{\Ll_0}^{(N)} \Big ( f_0 \prod_{j\in J} \partial_{\lambda_j} f_j \Big ),
\end{equation}
where $J \subseteq \{1,2,\ldots,d\}$ is a sub-set of indices. We defer the physical interpretation of these cyclic cocycles to the next section.

\section{The algebra of many-body Hamiltonians}
\label{Sec:PhysHam}

The first part of this section focuses on physical aspects, specifically on the dynamics of a gas of fermions populating a Delone set $\Ll$. The algebra of physical observables is that of canonical anticommutation relations (CAR) over the Hilbert space $\ell^2(\Ll)$. The time evolutions of the physical observables, which can be experimentally mapped and quantified, are generated by particular derivations that are analyzed here. These derivations supply the core algebra of physical Hamiltonians and the main challenge is to complete this and related algebras to $C^\ast$- and pro-$C^\ast$-algebras. This was accomplished in \cite{MeslandArxiv2021} via the essential extensions of the groupoid $C^\ast$-algebras introduced in the previous section. This solution as well as other related aspects are discussed in the second part of the section.

\subsection{The algebra of local observables}\label{Sec:ALO} The algebra of local observables for fermions hopping over a  Delone set $\Ll$ is the universal $C^\ast$-algebra ${\rm CAR}(\Ll)$ of the anti-commutation relations over the Hilbert space $\ell^2(\Ll)$ \cite{BratteliBook2}:
\begin{equation}\label{Eq:CAR}
a_x a_{x'} + a_{x'} a_x =0, \quad a_x^\ast a_{x'} + a_{x'}a_x^\ast = \delta_{x,x'}\cdot 1 , \quad x,x' \in \Ll.
\end{equation}
This algebra accepts many useful presentations, but here we will be interested in the following symmetric presentation:

\begin{proposition}{For $\xi \in \widehat {\rm Del}_{(r,R)}^{(n)}(\RM^d)$, let
\begin{equation}
\mathfrak a(\xi) :=a_{\chi_\xi(n)} \cdots a_{\chi_\xi(1)} \in {\rm CAR}(\Ll). 
\end{equation}
Then any element from ${\rm CAR}(\Ll)$ accepts a unique presentation as a convergent sum of the type}
 \begin{equation}\label{Eq:SymPres1}
 A = \sum_{n,m \in \NM} \tfrac{1}{\sqrt{n!m!}} \sum_{\zeta \in \mathfrak a_n^{-1}(\Ll)} \ \sum_{\xi \in \mathfrak a_m^{-1}(\Ll)} c(\zeta,\xi) \, \mathfrak a(\zeta)^\ast \mathfrak a(\xi),
 \end{equation}
 where the coefficients are bi-equivariant in the sense
 \begin{equation}
 c\big (\Lambda_{s_1}(\zeta), \Lambda_{s_2}^{-1}(\xi)\big) = (-1)^{s_1}c(\zeta,\xi) (-1)^{s_2}, \quad s_1 \in \Ss_n, \ s_2 \in \Ss_m.
 \end{equation}
 \end{proposition}
 
 \begin{remark}{\rm The symmetric presentation puts all possible orderings of the generators on equal footing and this is desirable when working with generic Delone sets or when deformations of the underlying lattice are allowed.
 }$\Diamond$
 \end{remark}
 
\begin{remark}{\rm The first sum in \eqref{Eq:SymPres1} includes the cases $n=0$ or $m=0$. In such situations, we use the convention that $\mathfrak a_0^{-1}(\Ll) = \emptyset$ and $\mathfrak a(\emptyset)=1$, the unit of ${\rm CAR}(\Ll)$. Therefore, $c(\emptyset,\emptyset)$ in \eqref{Eq:SymPres1} is the coefficient corresponding to the unit.
}$\Diamond$
\end{remark}
 
The ${\rm GICAR}(\Ll)$ subalgebra consists of the elements that are invariant against the $U(1)$-twist of the generators $a_x \mapsto \lambda a_x$, $\lambda \in \SM^1$, hence the name of gauge invariant elements. The symmetric presentation of such elements takes the form
\begin{equation}
 A = \sum_{n \in \NM} \tfrac{1}{n!} \sum_{\zeta,\xi \in \mathfrak a_n^{-1}(\Ll)} \ c(\zeta,\xi) \, \mathfrak a(\zeta)^\ast \mathfrak a(\xi),
 \end{equation}
 
The interest, however, is not in ${\rm CAR}(\Ll)$ or ${\rm GICAR}(\Ll)$, but in the dynamics of the local observables, i.e. in group homomorphisms $\alpha : \RM \to {\rm Aut}\big({\rm CAR}(\Ll)\big)$  and their generators. The physical constraints on these generators will be described next.

\subsection{The physical Hamiltonians}\label{Sec:MConclusions} We will be dealing exclusively with dynamics that conserves the fermion number, hence with time evolutions that descend on the ${\rm GICAR}$ subalgebra. Specifically:

 \begin{definition} A Galilean and gauge invariant Hamiltonian with finite interaction range is a correspondence
\begin{equation}\label{Eq:HCorresp}
{\rm Del}^r_R(\RM^d) \ni \Ll \mapsto H_\Ll = \sum_{n \in \NM^\times} \tfrac{1}{n!} \sum_{\zeta,\xi \in \mathfrak a_n^{-1}(\Ll)}  h_n(\zeta,\xi) \, \mathfrak a^\ast(\zeta) \mathfrak a(\xi),
\end{equation}
where the coefficients $h$ are continuous $\CM$-valued maps over the many-body covers obeying the constraint  $h_n(\zeta,\xi) = \overline{h_n(\xi,\zeta)}$ and:
\begin{enumerate}[\quad c1.]
\item $h_n$'s are bi-equivariant w.r.t. $\Ss_n$;
\item $h_n$'s are equivariant w.r.t. shifts, $h_n\big(\hat{\mathfrak t}_x(\zeta,\xi)\big ) = h_n(\zeta,\xi)$, $x \in \RM^d$;
\item $h_n$'s vanish whenever the diameter of $V_\zeta  \cup V_\xi$ exceeds a fixed value ${\rm R}_{\rm i}$.
\end{enumerate}
\end{definition}

\begin{remark}{\rm It is important to stress that {\it c2-3} are motivated by the experimental realities. For example, $c2$ is a consequence of the fact that the physical processes involved in the coupling of the quantum resonators are Galilean invariant. Also, $c3$ relates to the simple fact that a finite team of experimenters can only probe the couplings (encoded in the $h_n$'s) between a finite number of sites. Any Hamiltonian that cannot be approximated by a finite range Hamiltonian is beyond the control of the experimental team. Continuity of the coefficients comes from a similar argument: The experimenters can only probe a finite number of lattices, hence they need to rely on extrapolations. As such, Hamiltonians with discontinuous coefficients are beyond their control.
}$\Diamond$
\end{remark}

The following statement gives substance to the formal series in Eq.~\eqref{Eq:HCorresp}:

\begin{proposition} Let $\{\Ll_k\}$ be a net of finite sets converging to $\Ll$ and let $H_{\Ll_k}$ be the truncation of $H_\Ll$ from Eq.~\eqref{Eq:HCorresp} to an element of ${\rm CAR}(\Ll_k)$. Then the map
$$
{\rm ad}_{H_\Ll}(A) : =  \lim_{k \rightarrow \infty} \imath [A,H_{\Ll_k}] , \quad A \in \Dd(\Ll): = \mathop{\cup}_k {\rm CAR}(\Ll_k),
$$
is a derivation that leaves $\Dd(\Ll)$ invariant. Furthermore,
\begin{equation}\label{Eq:PhysDer}
{\rm ad}_{H_\Ll}(A)^\ast : = {\rm ad}_{H_\Ll}(A^\ast), \quad \forall \ A \in \Dd(\Ll),
\end{equation}
and ${\rm ad}_{H_\Ll}$ is closable and in fact a pre-generator of a time evolution.
\end{proposition}

\begin{remark}{\rm Obviously, we are dealing with outer derivations that are inner-limit \cite{BratteliBook1}. Such derivations leave the ideals of ${\rm GICAR}(\Ll)$ invariant, hence they descend on the quotients. This together with the fact that ${\rm GICAR}(\Ll)$ is a solvable $C^\ast$-algebra were the essential ingredients for the characterization of these derivations supplied in \cite{MeslandArxiv2021} and summarized in Theorem~\ref{Th:PunchLine}. 
}$\Diamond$
\end{remark}

All derivations~\eqref{Eq:PhysDer} leave their common domain invariant, hence they can be composed among themselves. This enables us to speak of:

\begin{definition}[\cite{MeslandArxiv2021}]\label{Def:CoreAlg} The core algebra of physical derivations is the sub-algebra $\dot \Sigma(\Ll) \subset {\rm End}\big (\Dd(\Ll)\big )$ generated by derivations ${\rm ad}_{Q_\Ll^n}$ corresponding to
\begin{equation}\label{Eq:Qn}
Q^n_\Ll = \tfrac{1}{n!}   \sum_{\zeta,\xi\in \mathfrak a^{-1}_{n}(\Ll)}
  q_{n} (\zeta,\xi) \mathfrak a^\ast(\zeta) \mathfrak a(\xi),
\end{equation}
where $q$'s are defined over the many-body covers and satisfy the constraints c1-3. Henceforth, an element of $\dot \Sigma (\Ll)$ can be presented as a finite sum
\begin{equation}\label{Eq:Qq}
\Qq = \sum_{\{Q\}} c_{\{Q\}} \ {\rm ad}_{Q_\Ll^{n_1}} \circ \ldots \circ {\rm ad}_{Q_\Ll^{n_k}}, \quad c_{\{Q\}} \in \CM.
\end{equation}
The multiplication is given by the composition of linear maps over $\Dd(\Ll)$.
\end{definition}

\begin{remark}{\rm The linear maps from Eq.~\eqref{Eq:Qq} can be also associated with formal sums as in Eq.~\eqref{Eq:HCorresp}, but the coefficients will no longer have finite range. The linear space spanned by ${\rm ad}_{Q_\Ll^n}$ is, however, invariant against the Lie bracket
\begin{equation}
({\rm ad}_{Q_\Ll^n},{\rm ad}_{Q_\Ll^m}) \mapsto {\rm ad}_{Q_\Ll^n} \circ {\rm ad}_{Q_\Ll^m} - {\rm ad}_{Q_\Ll^m} \circ {\rm ad}_{Q_\Ll^{n}}.
\end{equation}
Indeed, the right side has again coefficients with finite range. Hence, $\dot \Sigma(\Ll)$ is the associative envelope of this Lie algebra.
}$\Diamond$
\end{remark}

The interest is in formalizing the whole algebra $\dot \Sigma(\Ll)$ but also in its representations on the sectors with finite fermions of the Fock space. 

\begin{definition} The vacuum state over ${\rm CAR}(\Ll)$ is defined as
\begin{equation}
\omega(A) = c(\emptyset,\emptyset),
\end{equation}
where $A \in {\rm CAR}(\Ll)$ is assumed to be in its symmetric presentation~\eqref{Eq:SymPres1}. The corresponding GNS representation, also known as the Fock representation, will be denoted by $\pi_\omega$ and the equivalence class of $\mathfrak a(\xi)^\ast$ for this representation will be denoted by $|\xi\rangle :=\mathfrak a(\xi)^\ast + {\rm Null}(\omega)$, $\xi \in \mathfrak a_n^{-1}(\Ll)$ for some $n \in \NM$. When restricted to ${\rm GICAR}$ subalgebra, the Fock representation breaks into a direct sum $\pi_\omega = \bigoplus_{N\in \NM} \pi_\omega^N$, where $\pi_\omega^N$ is carried by the Hilbert space $\Ff_N^{(-)}$ linearly spanned by $|\xi\rangle$, $\xi \in \mathfrak a_N^{-1}(\Ll)$.
\end{definition}

The Fock representation extends over the core algebra $\dot \Sigma(\Ll)$. More precisely:

\begin{proposition}[\cite{MeslandArxiv2021}]\label{Pro:Vaccum} The Fock representation supplies representations of $\dot \Sigma(\Ll)$ by bounded operators on the $N$-fermion sectors $\Ff_N^{(-)}$ of the Fock space. Explicitly:
\begin{enumerate}[{\ \rm 1)}]
\item If $n >N$, then
\begin{equation*}
\pi_\eta^N(Q_\Ll^n) =0.
\end{equation*} 

\item If $n=N$, then
\begin{equation}\label{Eq:GoodRep}
\pi_\eta^N(Q_\Ll^N) = \tfrac{1}{N!}   \sum_{\zeta,\xi \in \mathfrak a_N^{-1}(\Ll)}
  q_{N}\big(\mathfrak t_{\chi_\xi(1)} (\zeta,\xi) \big ) |\zeta \rangle \langle \xi |.
\end{equation}

\item If $n<N$, then
\begin{equation*}
\pi_\eta^N(Q_\Ll^n) = \tfrac{1}{n!(N-n)!}   \sum_{\zeta,\xi\in \mathfrak a^{-1}_{n}(\Ll)} \ \sum_{\gamma \in \mathfrak a^{-1}_{N-n}(\Ll)}
  q_{n} (\xi,\zeta) |\xi \vee \gamma \rangle \langle \zeta \vee \gamma |.
\end{equation*}
\end{enumerate}
Above, $|\xi \vee \gamma \rangle$ is the equivalence class of $\mathfrak a(\xi)^\ast \mathfrak a(\gamma)^\ast$ in the GNS representation corresponding to $\omega$.
\end{proposition}

We point to the similarity of the above representation in the case $n=N$ and the left regular representations~\eqref{Eq:PiL} of the groupoid algebra over the $N$-body covers. There are two exceptional factors at work here, which enabled us in \cite{MeslandArxiv2021} to ultimately link the two representations. Firstly, $\Ff_N^{(-)}$ is isomorphic to the Hilbert space $\ell^2\big(\mathfrak a_N^{-1}(\Ll)\big )/\Nn$ appearing in Proposition~\ref{Pro:FockDescent}, with the isomorphism supplied by $|\xi\rangle \! \rangle + \Nn \mapsto |\xi \rangle$. Secondly, the sub-space of $\pi_\omega\big(\dot \Sigma(\Ll)\big )$ spanned by $\pi_\omega(Q_\Ll^N )$, with $Q_\Ll^N$ as in~ \eqref{Eq:Qn}, is an ideal of $\pi_\omega^N\big(\dot \Sigma(\Ll)\big )$. We now can state the main conclusion of \cite{MeslandArxiv2021}:

\begin{theorem}[\cite{MeslandArxiv2021}]\label{Th:PunchLine} Let $\Ll \in \Xi_{\Ll_0}$. Then:
\begin{enumerate}[\  \rm i)]
\item All representations of the physical Hamiltonians on the finite sectors of the Fock space can be generated from a left-regular representation of an essential extension of the bi-equivariant groupoid algebras. More precisely, for any $N \in \NM^\times$, there exists an embedding
\begin{equation}
\pi_\omega^N\big (\dot \Sigma(\Ll)\big ) \rightarrowtail \pi_\eta \Big (\Mm\big (C^\ast_{r,\Ss_N}\big(\widehat \Gg_{\Ll_0}^{(N)},\CM\big)\big ) \Big ),
\end{equation}
where $\pi_\eta$ is the left regular representation corresponding to any $\eta \in \mathfrak a_N^{-1}(\Ll)$ and $\Mm$ indicates the extension to multiplier algebra. 
\item The algebra $\dot \Sigma(\Ll)$ itself accepts a completion to a pro-$C^\ast$-algebra:
$$
\dot \Sigma(\Ll) \rightarrowtail \varprojlim \  \bigoplus_{n=1}^N \pi_\eta \Big (\Mm\big (C^\ast_{r,\Ss_n}\big(\widehat \Gg_{\Ll_0}^{(n)},\CM\big)\big ) \Big ) \otimes \pi_\eta \Big (\Mm\big (C^\ast_{r,\Ss_n}\big(\widehat \Gg_{\Ll_0}^{(n)},\CM\big)\big ) \Big )^{\rm op}. 
$$
\end{enumerate}
\end{theorem}

\begin{remark}{\rm Given the above, the pro-$C^\ast$-algebra
\begin{equation}
\Sigma(\Ll_0) : = \varprojlim \  \bigoplus_{n=1}^N \Mm\big (C^\ast_{r,\Ss_n}\big(\widehat \Gg_{\Ll_0}^{(n)},\CM\big)\big )  \otimes \Mm\big (C^\ast_{r,\Ss_n}\big(\widehat \Gg_{\Ll_0}^{(n)},\CM\big)\big ) ^{\rm op}
\end{equation}
can be seen as the sought completion of the core algebra of Hamiltonians. For this reason, we will refer to $\Sigma(\Ll_0)$ as the algebra of physical Hamiltonians over the Delone set $\Ll_0$.
}$\Diamond$
\end{remark} 

\begin{remark}{\rm There are direct physical interpretations of the essential ideal $C^\ast_{r,\Ss_N}\big(\widehat \Gg_{\Ll_0}^{(N)},\CM\big)$ and of the corona of the multiplier algebra: Under a time evolution generated by Hamiltonians from the algebra $C^\ast_{r,\Ss_N}\big(\widehat \Gg_{\Ll_0}^{(N)},\CM\big)$, the $N$ fermions evolve in a single cluster while, for a dynamics generated by Hamiltonians from the corona of the multiplier algebra, the fermions can scatter into two or more clusters.
}$\Diamond$
\end{remark}

\subsection{Transport coefficients} The anticommutation relations~\eqref{Eq:CAR} are invariant against the $U(1)$-gauge transformations $a_x \mapsto e^{\imath \langle k, x\rangle} a_x$, $k \in \RM$. Then the universality of the ${\rm CAR}$-algebra assures us that this transformation is in fact an automorphism of ${\rm CAR}(\Ll)$. Any automorphism can be lifted to inner-limit derivations, in particular to the ones generating the core algebra $\dot \Sigma(\Ll)$. It is straightforward to verify that these $U(1)$-gauge transformations coincide with the natural $\RM^d$-action over $C^\ast_{r,\Ss_N}\big(\widehat \Gg_{\Ll_0}^{(N)},\CM\big)$ introduced in Proposition~\ref{Pro:NRD}.

\vspace{0.2cm}

To clarify the physical meaning of the $\RM^d$-actions and of the associated derivations, we recall two important macroscopic physical observables, one corresponding to the fermion density $\mathfrak n$ and one to the fermion position $\mathfrak X$,
\begin{equation}
\mathfrak n = \sum_{x \in \Ll} a_x^\ast a_x, \quad \mathfrak X = \sum_{x \in \Ll} x \, a_x^\ast a_x,
\end{equation}
which can be defined for any $\Ll$ from ${\rm Del}_R^r(\RM^d)$. Note that $\mathfrak n$ belongs to the core algebra $\dot \Sigma(\Ll)$, but $\mathfrak X$ does not. Yet, we have
\begin{equation}
\partial_j A = \imath [\mathfrak X_j,A], \quad A \in {\rm CAR}(\Ll).
\end{equation}

\vspace{0.2cm}

We now consider $h \in C^\infty_{\Ss_N}\big(\widehat \Gg_{\Ll_0}^{(N)},\CM\big)$ and the corresponding Hamiltonian
\begin{equation}
H_\Ll = \tfrac{1}{N!}   \sum_{\zeta,\xi \in \mathfrak a_N^{-1}(\Ll)}
  h\big(\mathfrak t_{\chi_\xi(1)} (\zeta,\xi) \big ) \mathfrak a^\ast(\zeta) \mathfrak a(\xi)
\end{equation}
on $\ell^2(\Ll)$ with $\Ll \in \Xi_{\Ll_0}$. The macroscopic physical observables corresponding to the particle current-density vector of components is defined as
 \begin{equation}
 \mathfrak J_j : = \frac{d \mathfrak X_j}{d t} = \imath [\mathfrak X_j,H_\Ll],
\end{equation}
and we observe that 
\begin{equation}
\pi_\omega^N(\mathfrak J_j) = \pi_\eta(\partial_j h), \quad j=1,\ldots,d,
\end{equation}
for any $\eta \in \mathfrak u^{-1}(\Ll)$.

\begin{remark}{\rm Some of the above attributes must have confused the reader. Indeed, we declared in subsection~\ref{Sec:ALO} that the algebra of physical observables is ${\rm CAR}(\Ll)$. The elements $\mathfrak n$ and $\mathfrak J$ do not belong to ${\rm CAR}(\Ll)$ and this is why we used the terminology macroscopic physical observables. Their measurement requires entirely different experimental setups capable of probing the samples over large, macroscopic regions. Note that, in order to evaluate expected values for these macroscopic observables, one actually needs a state on the algebra $\dot \Sigma(\Ll)$. There is, of course, a dynamics on this algebra, generated by the inner derivation $\imath [\cdot , H_\Ll]$.
}$\Diamond$
\end{remark}

Now, let $p$ be a spectral projection of $h$, both seen as elements of $C^\ast_{r,\Ss_N}\big(\widehat \Gg_{\Ll_0}^{(N)},\CM\big)$. If $p$ belongs to the domain of the trace, then $f \mapsto \widehat \Tt_{\Ll_0}^{(N)}(f p)$ defines a state over $C^\ast_{r,\Ss_N}\big(\widehat \Gg_{\Ll_0}^{(N)},\CM\big)$, which, like any other state, can be uniquely extended over the multiplier algebra, hence over $\pi_\omega^N\big ( \dot \Sigma(\Ll)\big )$. In particular, we can evaluate the expected value of the fermion density $\mathfrak n$ and the result will be a non-zero value. Hence, 
\begin{equation}\label{Eq:TheState}
\dot \Sigma(\Ll) \ni Q \mapsto \widehat \Tt_{\Ll_0}^{(N)}\big (\pi_\omega^N(Q) \, \pi_\omega^N(p)\big)
\end{equation} 
generates a state at finite fermion density over the core algebra  $\dot \Sigma(\Ll)$, which is invariant against the dynamics. Above, $\widehat\Tt_{\Ll_0}^{(N)}$ is evaluated via the relation~\eqref{Eq:TrV}. The state~\eqref{Eq:TheState} corresponds to a uniform population of the energy spectrum covered by $p$, something that is routinely achieved in cold atom experiments, for example. Such states can be definitely achieved in electron systems if interest exists.

\vspace{0.2cm}

We can attach numerical invariants to these spectral states, by pairing the even cyclic cocycles~\eqref{Eq:CCoC} with the $K$-theoretic class of the spectral projection:
\begin{equation}\label{Eq:ChJ}
{\rm Ch}_J(p) = \Lambda_J \sum_{\lambda \in \Ss_J} (-1)^\lambda  \ \widehat\Tt_{\Ll_0}^{(N)} \Big ( p \prod_{j=1}^{|J|} \partial_{\lambda_j} p \Big ).
\end{equation}
They have a direct connection with various transport coefficients, which are all computed from the current correlation functions
\begin{equation}
\big \langle \mathfrak J_{j_1} \ldots \mathfrak J_{j_k} \big \rangle = \widehat \Tt_{\Ll_0}^{(N)}\big ((\partial_{j_1} h) \ldots (\partial_{j_k} h) \big ).
\end{equation}
By repeating the arguments from \cite{BellissardJMP1994,BellissardLNP2003}, which are purely algebraic in nature, the cocycles with $|J|=2$ can be given the physical interpretation of the linear transport coefficient for the plane corresponding to the indices $J$. Furthermore, while we have not covered yet the presence of a magnetic field, we are confident that, once that is done, the calculation from Section~5.6 of \cite{ProdanSpringer2016} can be repeated and the cocycles with $|J|\geq 4$ can be identified with specific non-linear transport coefficients.

\vspace{0.2cm}

There is definitely a whole lot to be learned about the physical interpretation of the states~\eqref{Eq:TheState}. Various strategies for completing the thermodynamic limit $N \to \infty$ have been listed in \cite{MeslandArxiv2021}. In particular, it was pointed out that, if a Hamiltonian is generated as above at a fix value of $N$, then such Hamiltonian lands in the corona of the multiplier algebra of $C^\ast_{r,\Ss_N}\big(\widehat \Gg_{\Ll_0}^{(N)},\CM\big)$, for larger $N$'s. However, such Hamiltonian can be pulled into $C^\ast_{r,\Ss_N}\big(\widehat \Gg_{\Ll_0}^{(N)},\CM\big)$ by sandwiching it between quasi-central approximations of the identity. These approximations can be crafted to control the fermion density and, essentially, generate finite-volume approximations that can be used to achieve thermodynamic limits with targeted fermion densities. Hence, the states~\eqref{Eq:TheState} are expected to be quite accurate for $N$ large enough. In this respect, we want to point out for the reader that the many-body Hamiltonians and associate states for various fractional Hall sequences constructed in \cite{ProdanPRB2009} are very close in spirit with the prescription  we just described.

\begin{remark}{\rm While the relation between the states~\eqref{Eq:TheState} and the equilibrium thermodynamic states is definitely of great importance, we want to point to the reader that the focus in materials science is firmly shifting from equilibrium to dynamical states. The activities related to the latter consists precisely in identifying and classifying Hilbert modules that are invariant under the dynamics, with the interest coming from a search and discovery of new manifestations of the so called bulk-defect correspondence principle \cite{ProdanJPA2021}. We believe that the framework supplied by the groupoid approach and the states~\eqref{Eq:TheState} designed with finite number of fermions will be an abundant source of interesting new examples, as already hinted by the numerical experiments from \cite{LiuPRB2022}.
}$\Diamond$
\end{remark}

\section{Index Theorems for the 1-Fermion Setting}

At this moment, we are not able to produce a complete proof of the index theorem for the pairings of the cocycles~\eqref{Eq:CCoC} and the K-theory of the corresponding $C^\ast$-algebras for the cases $N>1$. However, as a preparation, we reworked the proofs of the index theorems from \cite{BourneJPA2018,BourneAHP2020} for the case $N=1$, by generalizing the geometric identities from periodic lattices \cite{ProdanJPA2013,ProdanJFA2016} to generic Delone sets. The parings~\eqref{Eq:ChJ} for generic fermion numbers are very close in spirit to the case $N=1$ and, in our opinion, an index theorem is now within reach if we follow the template presented here. We recall that the challenge here is to validate the proofs in Sobolev regimes. We will showcase only the case of even dimensions because the case of odd dimensions is very similar.

\vspace{0.2cm}

Since we retreat to the case $N=1$, hence to the groupoid $\Gg_{\Ll_0}$ introduced in subsection~\ref{Sec:CanonicalG}, we can simplify the notation a bit. Indeed, in this case, the left regular representations are indexed by $\Ll \in \Xi_{\Ll_0}$, hence they will be denoted by $\pi_\Ll$. The semi-finite trace $\widehat \Tt_{\Ll_0}^{(N)}$ becomes a continuous trace, which will be denoted by $\Tt$.

\subsection{Dirac operators on Delone sets} We introduce here a class of Dirac operators over a point set $\Ll \in \Omega_{\Ll_0}$, with $\Ll_0$ fixed in ${\rm Del}_R^r(\RM^d)$, and establish some of their basic properties. For this, let $\gamma_1, \ldots, \gamma_{d}$ be an irreducible representation of the complex Clifford algebra on $\CM^n$, $n=2^\frac{d}{2}$:
\begin{equation}
\gamma_i \gamma_j + \gamma_j \gamma_i = 2 \delta_{ij}, \quad i,j=1,\ldots,d.
\end{equation}
When $d$ is even, the Clifford algebra accepts the grading $\gamma_0=-\imath^{-n}\gamma_1 \ldots \gamma_{d}$. We denote by $\mathrm{tr}_\gamma$ the ordinary trace over $\CM^n$ and let
\begin{equation}
\Hh_\Ll=\CM^n \otimes \ell^2(\Ll),
\end{equation}
which is the Hilbert space which will carry the index theorems.
We will also use the same symbol $\mathrm{Tr}$ for the ordinary trace over the Hilbert space $\Hh_\Ll$. 

\begin{definition} For $\Ll \in \Omega_{\Ll_0}$, we call $\delta_{\Ll} : \Ll \to \RM^d$ a set of acceptable shifts if the set $\{\delta_\Ll(x), \ x \in \Ll\}$ is bounded in $\RM^d$ and $x - \delta_\Ll(x) \neq 0$ for all $x \in \Ll$.
\end{definition} 

\begin{definition} The Dirac operator $D_{\Ll,\delta_\Ll}$ associated to $\Ll \in \Omega_{\Ll_0}$ and a set $\delta_\Ll$ of acceptable shifts is the unbounded self-adjoint operator
\begin{equation}
D_{\Ll,\delta_\Ll} : =\sum_{j=1}^{d} \sum_{x \in \Ll} \gamma_j \otimes \big (x-\delta_\Ll(x)\big )_j \, |x\rangle \langle x|.
\end{equation}
\end{definition}

The spectrum of the operator is
\begin{equation}
{\rm Spec}(D_{\Ll,\delta_\Ll}) = \big \{\pm|x-\delta_\Ll(x)|,\ x \in \Ll \big \}
\end{equation}
and, by construction, $0\notin {\rm Spec}(D_{\Ll,\delta_\Ll})$. This will enable us to define the phase of the Dirac operator but the consideration of the shifts has a much deeper motivation, revealed in the proof of the main index theorem. Let us add that, under the natural unitary map between $\ell^2(\Ll)$ and $\ell^2\big(\mathfrak t_y(\Ll)\big )$, $y \in \RM^d$, we have
\begin{equation}\label{Eq:UD}
D_{\Ll,\delta_\Ll} \mapsto D_{\mathfrak t_y(\Ll),\mathfrak t_y \circ \delta_\Ll \circ \mathfrak t_y^{-1}}.
\end{equation}
 and note that $\mathfrak t_y \circ \delta_\Ll \circ \mathfrak t_y^{-1}$ is again an acceptable set of shifts for $\mathfrak t_y (\Ll)$.

\begin{notation}{\rm Given a metric space $(Y,|\cdot |)$ and $y \in Y$, we will consistently use the notation $\hat y : = y/|y|$ if $|y| \neq 0$ and $\hat y : =0$ otherwise. Furthermore, we introduce the shorthand $\sum_{j=1}^d \gamma_j \otimes v_j = \gamma \cdot v$, for any $d$-tuple $v$ from a linear space.
}$\Diamond$
\end{notation}

\begin{definition} Let $\Ll \in \Omega_{\Ll_0}$ and $\delta_\Ll$ an acceptable set of shifts. Then, we define the phase of the corresponding Dirac operator as
\begin{equation}
\hat D_{\Ll,\delta_\Ll} = D_{\Ll,\delta_\Ll}/\sqrt{D^2_{\Ll,\delta_\Ll}}.
\end{equation} 
\end{definition}

Before stating the standard properties of the Dirac operators and their phases, we mention the following simple asymptotic estimate, which will come handy on several occasions:

\begin{proposition} For any $x \in \RM^d$ with $|x|=1$ and $s \in \RM$, we have 
\begin{equation}\label{Eq:Assy0}
\lim_{s \to \infty} s \big(\widehat{s x+y} - \widehat{sx}\big)  = y -\tfrac{1}{2} \langle x,  y \rangle \, x. 
\end{equation}
\end{proposition}

\begin{proposition}\label{Pro:HadD1} Let $\Ll \in \Omega_{\Ll_0}$ and $\delta_\Ll$ and $\delta'_\Ll$ be any two acceptable sets of shifts. Then: 
\begin{enumerate}[\rm i)]
\item $\hat D_{\Ll,\delta_\Ll} - \hat D_{\Ll,\delta'_\Ll}$ is a compact operator.
\item $|D_{\Ll,\delta_\Ll}|^{-d} - |D_{\Ll,\delta'_\Ll}|^{-d}$ is a trace class operator.
\end{enumerate}
\end{proposition}

\begin{proof} The statements follow from elementary estimates involving the asymptotic expression from Eq.~\eqref{Eq:Assy0} (see {\it e.g.} \cite{ProdanSpringer2016} or \cite{BourneAHP2020}).
\end{proof}

\begin{definition} Let $\Ll \in \Omega_{\Ll_0}$ and $B \in \BM\big(\ell^2(\Ll)\big)$, the algebra of bounded operators over $\ell^2(\Ll)$. We say that $B$ has finite range if there exists $M \in \RM_+$ such that
\begin{equation}
\langle x | B |x' \rangle =0, \quad \forall \ |x-x'| > M.
\end{equation}
\end{definition}

\begin{proposition}\label{Pro:CompactD} Let $\Ll \in \Omega_{\Ll_0}$ and $A \in \BM\big(\ell^2\big (\Ll)\big )$ such that $A$ is the limit in $\BM(\ell^2\big (\Ll)\big )$ of a sequence of operators with finite range.  Then $[\hat D_{\Ll,\delta_\Ll},A]$ is a compact operator.
\end{proposition} 

\begin{proof} We consider first an operator $B$ of finite range. Then,
\begin{equation}\label{Eq:X345}
\langle x|[\hat D_{\Ll,\delta_\Ll},B]|x'\rangle = \gamma \cdot  \big(\widehat{x+\delta_\Ll(x)} - \widehat{x'+\delta_\Ll(x')}\big) \langle x | B |x' \rangle,
\end{equation}
hence $[\hat D_{\Ll,\delta_\Ll},B]$ has finite range. Furthermore, using the asymptotic behavior described in Eq.~\eqref{Eq:Assy0}, we can concluded that $\langle x|[\hat D_{\Ll,\delta_\Ll},B|x'\rangle|$ converges to zero as $|x| \to \infty$, uniformly in $x'$ (because $B$ has finite range). This assures us that $[\hat D_{\Ll,\delta_\Ll},B]$ is a compact operator. Now, if $A$ is generic but obeys the stated assumptions, then $[\hat D_{\Ll,\delta_\Ll},A]$ can be approximated in norm by $[\hat D_{\Ll,\delta_\Ll},B_n]$, where $B_n$ is a sequence in $\BM\big(\ell^2\big (\Ll)\big )$ of finite range operators. Thus $[\hat D_{\Ll,\delta_\Ll},A]$ is a limit of compact operators, hence compact.
\end{proof}

\subsection{Dixmier trace} While singular traces are essential tools in the modern index theory \cite{ConnesGAFA1995}, in the present approach, they will only be used to establish the summability of the Fredholm modules, following \cite{BellissardJMP1994} and \cite{ProdanJPA2013} as models. A comprehensive treatment of singular traces over $\KM\big(\ell^2(\Ll)\big )$ can be found in \cite{LordGruyterBook}. Here, we briefly recall that Dixmier trace involves the algebra $\KM\big(\ell^2(\Ll)\big )$ of compact operators over $\ell^2(\Ll)$. Its unitization, as any other unital AF-algebra, can be pushed through Stratila and Voiculescu machinery \cite{StratilaSpringer1975}, to construct its m.a.s.a (maximal abelian subalgebra), which is just $C(\alpha \Ll)$, the algebra of continuous functions over the one-point compactification of $\Ll$. The associated subgroup $\Gamma$, as defined in \cite{StratilaSpringer1975}, is the subgroup of homeomorphisms of $\alpha \Ll$ sending $\infty$ to itself. Furthermore, the conditional expectation $P $ from $\CM \cdot I \oplus \KM\big(\ell^2(\Ll)\big )$ to its m.a.s.a $C(\alpha \Ll)$ is
\begin{equation}
\mathfrak P(c I+K) = G, \quad G(x) = c + \langle x |K|x\rangle.
\end{equation}
Traces on the positive cone of $\KM\big(\ell^2(\Ll)\big )$ can be generated from $\Gamma$-invariant additive maps $\tau$ on the positive cone of $C(\alpha \Ll)$, via the composition $\tau \circ \mathfrak P$. In fact, according to \cite{StratilaSpringer1975}, all trace states on a unital AF-algebra come in this way.

\vspace{0.2cm}

$\Gamma$-invariant maps can be constructed from weighted partial sums, such as
\begin{equation}\label{Eq:PSum}
s_N(f) = \frac{1}{C_N} \sum_{x \in \Ll} f(x) \mathfrak c_N(x), \quad f : \Ll \to \RM_+,
\end{equation}
where $\mathfrak c_N$ is the indicator function of the $d$-dimensional cube of size $N$ and centered at the origin. The coefficients $C_N$ can be essentially supplied by any monotone increasing function from $\NM$ to $\RM_+$. Such weighted sums supply maps $s(f): \NM \to [0,\infty]$ and any ultrafilter $\omega$ on $\NM$ supplies an ultrafilter $s(f)_\ast \omega$, the push forward ultrafilter, on the compact Hausdorff space $[0,\infty]$. Necessarily, $s(f)_\ast \omega$ has a unique limit in $[0,\infty]$ and, as such, one can define
\begin{equation}\label{Eq:ULimit}
\tau_\omega (f) : = \lim s(f)_\ast \omega.
\end{equation}
This map is additive and, furthermore, if $\omega$ is a free filter, then the limit in~\eqref{Eq:ULimit} coincides with $\lim_{N \to \infty} s_N(f)$, whenever the latter exists.

\begin{definition} The Dixmier trace ${\rm Tr}_{\rm Dix}$ over the positive cone of $\KM\big (\ell^2(\Ll)\big )$ is defined as the map $\tau_\omega \circ P$ with $\tau_\omega$ as above and with the choice 
\begin{equation}
C_N = \ln \sum_{x \in \Ll} \mathfrak c_N(x)
\end{equation}
in Eq.~\eqref{Eq:PSum} and $\omega$ assumed to be a free ultrafilter.
\end{definition}

\begin{remark}{\rm We have included the above digression because we developed an interest in singular traces over ideals of generic AF-algebras and it seems to us that the machinery invented by Stratila and Voiculescu \cite{StratilaSpringer1975} supplies the right environment and the tools for studying such singular traces.
}$\Diamond$
\end{remark} 

The statement involving the Dixmier trace, of great value for us, is:
\begin{lemma}\label{Lem:Dix1} If $g :\Xi_{\Ll_0} \to [0,\infty]$ is integrable, $\Ll \in \Xi_{\Ll_0}$ and $\varphi : \Ll \to \RM_+$ is uniformly bounded, then, $\PM$-almost surely, any positive compact operator $B$ over $\ell^2(\Ll)$ with diagonal
\begin{equation}\label{Eq:Diagonal}
\mathfrak P(B)(x) =  \varphi(x) g\big (\mathfrak t_x (\Ll)\big ) |x + \delta_\Ll(x) |^{-d}
\end{equation}
belongs to the domain of the Dixmier trace. In particular, it is trace class.
\end{lemma}

\begin{proof} From \cite{BourneJPA2018}[Lemma~6.1] (see also \cite{AzamovArxiv2022}[Th.~1.2]), one knows that $G(x) := g\big (\mathfrak t_x (\Ll)\big ) |x + \delta_\Ll(x) |^{-d}$ is $\PM$-almost surely in the domain of $\tau_\omega$, independently of the choice of the ultrafilter $\omega$. Furthermore,
\begin{equation}
\tau_\omega(G) = |\SM^{d-1}|^{-1} \int_{\Xi_{\Ll_0}} {\rm d}\PM(\Ll) \, g(\Ll).
\end{equation}
The statement follows because $\varphi$ is assumed uniformly bounded.
\end{proof}

\subsection{The Sobolev space}\label{Sec:Sobolev} The $L^s$-spaces relative to the trace $\Tt$, denoted here by $L^s(\Gg_{\Ll_0},\Tt)$, are defined as the completion of the convolution algebra $C_c(\Gg_{\Ll_0})$ under the norms
\begin{equation}
\|f\|_s = \Tt\big ( |f|^s \big )^\frac{1}{s}.
\end{equation}
The space $L^2(\Gg_{\Ll_0},\Tt)$ is the Hilbert space carrying the GNS representation of $C^\ast_r(\Gg)$ induced by $\Tt$ and $L^\infty(\Gg_{\Ll_0},\Tt)$ is the von Neumann algebra corresponding to the weak closure of $C^\ast_r(\Gg)$ in this representation. We recall that von Neumann algebras are stable under the Borel calculus, hence $\phi(f) \in  L^\infty(\Gg_{\Ll_0},\Tt)$ for any $f \in C^\ast_r(\Gg_{\Ll_0})$ and Borel function $\phi$ over the real line. In particular, the spectral projections of any self-adjoint element $h \in C^\ast_r(\Gg_{\Ll_0})$ belong to $L^\infty(\Gg_{\Ll_0},\Tt)$.

The Sobolev space that we need for the index theorem is supplied by the completion of $C_c(\Gg_{\Ll_0})$ under the norm
\begin{equation}\label{Eq:SNorm}
\|f\|_W = \|f\|_\infty + \sum_{i=1}^d \|\partial_i f\|_d.
\end{equation}
It will be denoted by $W(\Gg_{\Ll_0},\Tt)$.

\begin{proposition} $W(\Gg,\Tt)$ is a locally convex algebra.
\end{proposition}

\begin{proof} $W(\Gg,\Tt)$ is already a locally convex topological vector space, hence we only need to verify that the product is jointly continuous. We have,
\begin{align}
\| f_1 f_2\|_W & = \|f_1 f_2 \|_\infty + \sum_i \|(\partial_i f_1) f_2 + f_1 (\partial_i f_2) \|_d \\
& \leq \|f_1\|_\infty \| f_2 \|_\infty + \sum_i \|(\partial_i f_1)\|_d \|f_2\|_\infty + \|f_1\|_\infty \| (\partial_i f_2) \|_d \\
& \leq \Big (\|f_1\|_\infty + \sum_{i=1,d} \|\partial_i f_1\|_d\Big )\Big (\|f_2\|_\infty + \sum_{i=1,d} \|\partial_i f_2\|_d\Big ),
\end{align}
where H\"older’s inequality was used at places. Thus, $\| f_1 f_2\|_W \leq \| f_1\|_W \|f_2\|_W$ and the statement follows.
\end{proof}

The following statements re-assess the cocycles~\eqref{Eq:CCoC} in this new context:

\begin{proposition} The relations
\begin{equation}
W(\Gg_{\Ll_0},\Tt)^{\times d} \ni (f_0,f_1,\ldots,f_d) \mapsto \Tt\big (f_0 (\partial_1 f_1) \, \ldots (\partial_d f_d)\big ).
\end{equation}
defines a jointly continuous multi-linear map over $W(\Gg_{\Ll_0},\Tt)$.
\end{proposition}

\begin{proof} We have
\begin{align}\label{Eq:X132}
\Big | \Tt\big (f_0 (\partial_1 f_1) \, \ldots (\partial_d f_d)\big ) & \leq \|f_0\|_\infty \|(\partial_1 f_1) \, \ldots (\partial_d f_d)\|_1 \\
& \leq \|f_0\|_\infty \prod_{i=1}^{d}\|(\partial_i f_i)\|_d,
\end{align}
where H\"older’s inequality was used at places. Therefore,
\begin{align} 
\big | \Tt\big (f_0 (\partial_1 f_1) \, \ldots, (\partial_d f_d)\big ) \big |\leq \|f_0\|_W \|f_1\|_W  \ldots \| f_d\|_W,
\end{align}
and the statement follows.
\end{proof}

\begin{proposition} Consider the cyclic and jointly continuous multi-linear map
\begin{equation}
\Sigma_d (f_0,f_1,\ldots,f_d) = \Lambda_d \sum_{\sigma \in \Ss_d}(-1)^\sigma \Tt\big (f_0 (\partial_1 f_{\sigma_1}) \, \ldots (\partial_d f_{\sigma_d})\big )
\end{equation}
over the Sobolev space. Then
\begin{equation}\label{Eq:CycleCont}
|\Sigma_d (f_0,f_1,\ldots,f_d) - \Sigma_d (f'_0,f'_1,\ldots,f'_d)| \leq {\rm ct.} \sum_{i=0}^d \big (\|f_i - f'_i\|_W - \|f_i - f'_i\|_\infty \big ),
\end{equation}
where the constant in front is determined by the Sobolev norms of the entries.
\end{proposition}

\begin{proof} We have
\begin{align}
& \Sigma_d (f_0,f_1,\ldots,f_d) - \Sigma_d (f'_0,f'_1,\ldots,f'_d) \\
& \qquad  = \sum_{k=0}^d \Sigma_d \big (f'_0, \ldots , f'_{k-1}, (f_k -f'_k),f_{k+1},\ldots, f_d\big ) \\
& \qquad  = \sum_{k=0}^d \Sigma_d \big (f_{k+1},\ldots, f_d,f'_0, \ldots , f'_{k-1}, (f_k -f'_k)\big ).
\end{align}
Then the statement follows from the estimate in Eq.~\eqref{Eq:X132}.
\end{proof}

\begin{remark}{\rm The estimate in Eq.~\eqref{Eq:CycleCont} communicates that the value of the cocycle varies continuously under deformations inside $W(\Gg,\Tt)$ that are continuous only w.r.t. the differential piece of the Sobolev norm, {\it i.e.} the second term in Eq.~\eqref{Eq:SNorm}. This observation is crucial because the deformations occurring in the applications are never continuous w.r.t. the full Sobolev norm. Note that this special property is facilitated by the multi-linear map being cyclic.
}$\Diamond$
\end{remark}

\subsection{The Fredholm module and its quantized pairing} As in the previous sections and subsections, we will fix a Delone set $\Ll_0$. The $C^\ast$-algebra $C^\ast_r(\Gg_{\Ll_0})$ can be represented on $\Hh_\Ll$ by $1 \otimes \pi_\Ll$ and, to simplify the notation, we will denote this representation by the same symbols $\pi_\Ll$. Here $\Ll$ is assumed from $\Xi_{\Ll_0}$. Likewise, the natural actions of $\gamma$'s on $\Hh_\Ll$ will be denoted by the same symbols. It is clear that $\pi_\Ll(f) \gamma_0=\gamma_0 \pi_\Ll(f)$, for any $f \in C^\ast_r(\Gg_{\Ll_0})$. As for the Dirac operators, we will restrict the acceptable shifts to uniform shifts $\delta_\Ll(x) = w$, $w \in \RM^d$ and simplify the notation $\hat D_{\Ll,\delta_\Ll}$ to $\hat D_{\Ll}(w)$.

\begin{proposition}\label{Pro:FredMod} The $\Ll$-dependent tuple 
$$
\Big ( W(\Gg_{\Ll_0},\Tt), \pi_\Ll : M_{n}(\CM) \otimes W(\Gg_{\Ll_0},\Tt) \to \BM\big(\Hh_\Ll \big), \hat D_{\Ll}(w),\gamma_0 \Big )
$$ 
is $\PM$-almost surely a $d+\epsilon$-summable even Fredholm module.
\end{proposition}

\begin{proof} We will prove the summability. We are going to evaluate the Dixmier trace of the positive operator 
$$
{\rm tr}_\gamma\, \Big (\imath[\hat{D}_{\Ll}(w),\pi_\Ll (f)]^d \Big ), \quad f \in W(\Gg_{\Ll_0},\Tt).
$$ 
After projecting on the diagonal part, using Eq.~\eqref{Eq:X345} as well as the asymptotic behavior from Eq.~\eqref{Eq:Assy0}, we find
\begin{equation}
\mathrm{Tr}_{\mathrm{Dix}}\Big \{{\rm tr}_\gamma\, \Big ( \imath[\hat{D}_{\Ll}(w),\pi_\Ll (f)]^d \Big ) \Big \} = \tau_\omega (G),
\end{equation}
where $G : \Ll \to \RM_+$ is given by
\begin{align}
G(x) & =  \mathrm{tr}_\gamma \Big (\langle 0 |\prod_{i=1}^d \gamma(\hat x) \cdot \imath[X,\pi_{\mathfrak{t}_{x}\Ll}(f)]|0 \rangle \Big )  |x+w|^{-d} \\
& =  \mathrm{tr}_\gamma \Big (\langle 0 |\prod_{i=1}^d \gamma(\hat x) \cdot \pi_{\mathfrak t_x \Ll}(\nabla f)] |0 \rangle \Big )  |x+w|^{-d},
\end{align}
with $\gamma_i(\hat x)=\gamma_i -\hat x_i (\hat x \cdot  \gamma)$.
We can see that $G(x)$ is covered by Lemma~\ref{Lem:Dix1}, whenever $f \in W(\Gg_{\Ll_0},\Tt)$.
\end{proof}

\begin{theorem} Let $p \in W(\Gg_{\Ll_0},\Tt)$ be a projection and let $\pi_\Ll^\pm$ be the restrictions of $\pi_\Ll$ representation on the sectors of the grading $\gamma_0$. Then, $\PM$-almost surely, the operator 
\begin{equation}\label{Eq:FO}
\pi_\Ll^-(p)\hat D_{\Ll}(w) \pi_\Ll^+(p)
\end{equation}
is in the Fredholm class. Furthermore, if the hull $\Omega_{\Ll_0}$ and its measure are invariant against the inversion operation $\Ll \mapsto -\Ll$, then, $\PM$-almost surely,
\begin{equation}\label{Eq:FredholmIndex}
\mathrm{Index}\left(\pi_\Ll^-(p)\hat D_{\Ll}(w) \pi_\Ll^+(p)\right)=\Lambda_d \sum_{\lambda \in \Ss_d}(-1)^\lambda \ {\mathcal T} \Big ( p \prod_{i=1}^{d} \partial_{\lambda_i}p \Big ),
\end{equation}
where $\Lambda_d = \frac{(2\imath \pi)^\frac{d}{2}}{(d/2)!}$.
\end{theorem}

\begin{proof} Proposition~\ref{Pro:FredMod} assures us that the operator~\eqref{Eq:FO} is indeed Fredholm, $\PM$-almost surely. We observe first that the index in Eq.~\eqref{Eq:FredholmIndex} is insensitive to the choice of $w \in \RM^d$, as long as the latter supplies an acceptable set of shifts, which is the case almost surely with respect to the Lebesgue measure of $\RM^d$. Next, we also observe that the index is $\PM$-almost surely insensitive to the choice of $\Ll$ from $\Xi_{\Ll_0}$. To prove the claim, it is enough to demonstrate that the index is invariant against the groupoid action on $\Xi_{\Ll_0}$, specifically, against replacing $\Ll$ by $\mathfrak t_y(\Ll)$, $y \in \Ll$. For this, we perform a unitary transformation from $\ell^2(\Ll)$ to $\ell^2\big ( \mathfrak t_y(x)\big )$ and use the statement from Eq.~\eqref{Eq:UD} to observe that the Fredholm operator transforms into
\begin{equation}
 \pi_{\mathfrak t_y(\Ll)}^-(p)\hat D_{\mathfrak t_y(\Ll)}(y+w) \pi_{\mathfrak t_y(\Ll)}^+(p).
\end{equation}
According to Proposition~\ref{Pro:HadD1}, this is a compact perturbation of 
\begin{equation}
 \pi_{\mathfrak t_y(\Ll)}^-(p)\hat D_{\mathfrak t_y(\Ll)}(w) \pi_{\mathfrak t_y(\Ll)}^+(p),
\end{equation}
hence the index in Eq.~\eqref{Eq:FredholmIndex} is indeed invariant against $\Ll \mapsto \mathfrak t_y(\Ll)$. 

We now start the computation of the index. Proposition~\ref{Pro:FredMod} assures us that, $\PM$-almost surely, the index is supplied by the pairing with the Connes-Chern character,
\begin{equation}
{\rm Ind} \big (\pi_\Ll^-(p)\hat D_{\Ll}(w) \pi_\Ll^+(p) \big ) = \tfrac{1}{2} \, {\rm Tr} \big \{ \gamma_0 \hat D_{\Ll}(w) [\hat D_{\Ll}(w), \pi_\Ll(p)]^{d+1} \big \}.  
\end{equation}
We express the trace as the absolutely convergent sum
\begin{equation}
\sum_{x \in \Ll} {\rm tr}_\gamma \big ( \gamma_0  \langle x | \hat D_{\Ll}(w) [\hat D_{\Ll}(w), \pi_\Ll(p)]^{d+1} |x\rangle\big )  
\end{equation}
and observe that, by using the elementary identities,
$$
\hat D_{\Ll}(w) [\hat D_{\Ll}(w), \pi_\Ll(p)] = -[\hat D_{\Ll}(w), \pi_\Ll(p)] \hat D_{\Ll}(w), \quad \hat D_{\Ll}(w)|x\rangle = |x\rangle \, (\gamma \cdot \widehat{x + w}),
$$
as well as the cyclic property of ${\rm tr}_\gamma$, the summand can be processed as
\begin{align*}
 {\rm tr} \big ( \gamma_0  \langle x | & \hat D_{\Ll}(w)  \big (\hat D_{\Ll}(w) \, \pi_\Ll(p) -\pi_\Ll(p) \, \hat D_{\Ll}(w)\big )[\hat D_{\Ll}(w), \pi_\Ll(p)]^d |x\rangle\big ) \\
& \qquad \qquad = 2\, {\rm tr} \big ( \gamma_0  \langle x | \pi_\Ll(p) [\hat D_{\Ll}(w), \pi_\Ll(p)]^d |x\rangle\big ).
\end{align*}
Therefore,
\begin{equation}
{\rm Ind} \big (\pi_\Ll^-(p)\hat D_{\Ll}(w) \pi_\Ll^+(p) \big ) =  \sum_{x \in \Ll} {\rm tr}_\gamma\big ( \gamma_0 \langle x |\pi_\Ll(p) [\hat D_{\Ll}(w), \pi_\Ll(p)]^d |x\rangle \big ),  
\end{equation}
and note that the sum remains absolutely convergent. We now insert resolutions of the identity and write the right side as
\begin{equation}\label{Eq:Inter5}
\begin{aligned}
  & \sum_{x \in \Ll} \sum_{x_i \in \Ll} \Big (\prod_{i=0}^{d}\langle x_i | \pi_\Ll(p)|x_{i+1} \rangle \Big )
 {\rm tr}_\gamma \Big ( \gamma_0 \prod_{i=1}^d \gamma \cdot \big ( \widehat{x_i +w} - \widehat{x_{i+1} +w}\big ) \Big ),
\end{aligned}
\end{equation}
where $x_0 = x_{d+1}=x$. Our next goal is to de-couple the sums over $x$ and over $x_i$'s by using the insensitivity to the choice of $\Ll$ or $w$. For this, it is convenient to introduce the new coordinates $y_i$ that live on $\mathfrak t_x(\Ll)$ and such that $x_i = x + y_i$. Then the right side of Eq.~\eqref{Eq:Inter5} becomes
\begin{equation}
\begin{aligned}
&  \sum_{x \in \Ll} \sum_{y_i \in {\mathfrak t_x(\Ll)}} \Big (\prod_{i=0}^{d}\langle y_i | \pi_{\mathfrak t_x(\Ll)}(p)|y_{i+1} \rangle \Big ) \\
& \qquad \qquad \qquad \qquad \qquad \times {\rm tr}_\gamma \Big ( \gamma_0 \prod_{i=1}^d \gamma \cdot \big ( \widehat{x+y_i +w} -  \widehat{x+y_{i+1} +w}\big )\Big ).
\end{aligned}
\end{equation}
where $y_0 = y_{d+1}=0$. Consider  a term in the above sum that corresponds to a particular $x \in \Ll$  and let $\Ll' = \mathfrak t_x (\Ll)$. Note that $y:= \mathfrak t_x(0) = -x$ belongs to $\Ll'$. Then such term can be completely written in term of the shifted lattice $\Ll'$, as
\begin{equation}
\begin{aligned}
&\sum_{y_i \in \Ll'} \Big (\prod_{i=0}^{d}\langle y_i | \pi_{\Ll'}(p)|y_{i+1} \rangle \Big )  {\rm tr}_\gamma \Big ( \gamma_0 \prod_{i=1}^d \gamma \cdot \big (\widehat{y_i +w-y} - \widehat{y_{i+1} +w-y}\big ).
\end{aligned}
\end{equation}
Note also that $y-w$ belongs to $\mathfrak t_w(\Ll')$. Thus, if we let $\Ll'$ sample the entire $\Xi_{\Ll_0}$ and take into account the invariance of the index, we can write
\begin{equation}
\begin{aligned}
 {\rm Ind} \big (\pi_\Ll^-(p) & \hat D_{\Ll}(w) \pi_\Ll^+(p) \big )  = \int_{\Xi_{\Ll_0}} {\rm d} \PM(\Ll')\sum_{y_i \in \Ll'} \Big (\prod_{i=0}^{d}\langle y_i | \pi_{\Ll'}(p)|y_{i+1} \rangle \Big ) \\
& \times \sum_{y \in \mathfrak t_{w} (\Ll')} {\rm tr}_\gamma \Big ( \gamma_0 \prod_{i=1}^d \gamma \cdot  \big (\widehat{y_i -y} - \widehat{y_{i+1} -y}\big )\Big ).
\end{aligned}
\end{equation}
By giving values to $w$, we can sample the entire hull $\Omega_{\Ll_0}$ in the second line and, at this point, we indeed achieved the decoupling we mentioned at the beginning:
\begin{equation}
\begin{aligned}
 {\rm Ind} \big (\pi_\Ll^-(p)\hat D_{\Ll}(w) & \pi_\Ll^+(p) \big )  = \int_{\Xi_{\Ll_0}} {\rm d} \PM(\Ll')\sum_{y_i \in \Ll'} \Big (\prod_{i=0}^{d}\langle y_i | \pi_{\Ll'}(p)|y_{i+1} \rangle \Big ) \\
& \times \int_{\Omega_{\Ll_0}} {\rm d} \bar \PM(\Ll'')\sum_{y \in \Ll''} {\rm tr}_\gamma \Big ( \gamma_0 \prod_{i=1}^d \gamma \cdot  \big (\widehat{y_i -y} - \widehat{y_{i+1} -y}\big ) \Big ).
\end{aligned}
\end{equation}
Then the statement follows from the geometric identity stated below.
\end{proof}

\begin{lemma}\label{Lem:Id2} Let $y_1,\ldots, y_{d+1}$ be points of $\mathbb{R}^{d}$ with $y_{d+1}={0}$. Then, if the hull $\Omega_{\Ll_0}$ and its measure are invariant against the inversion operation $\Ll \mapsto -\Ll$, the following identity holds:
\begin{eqnarray}\label{Identity}
\int\limits_{\Omega_{\Ll_0}}{\rm d}\bar \PM(\Ll) \ \sum_{w \in \Ll} \mathrm{tr}_\gamma \Big (\gamma_0 \prod_{i=1}^{d}\gamma \cdot \left(\widehat{y_i-w} -\widehat{y_{i+1}-w} \right )\Big ) 
 =\Lambda_d \sum_{\lambda} (-1)^\lambda \prod_{i=1}^{d} (y_i)_{\lambda_i}.
\end{eqnarray}
\end{lemma}

\begin{proof} As in \cite{ProdanJPA2013}, the proof relies on a geometric interpretation of the trace
\begin{equation}\label{GeomId}
\mathrm{tr}_\gamma\big (\gamma_0 (\gamma \cdot y_1) \cdots (\gamma \cdot y_d)\}=\alpha_d \  \mathrm{Vol}\big[{0},{y}_1,\ldots,{y}_{d}\big], \quad \alpha_d = (2\imath)^\frac{d}{2} d!,
\end{equation}
where $[{y}_0,{y}_1,\ldots,{y}_{d}]$ is the simplex with vertices at ${y}_0$, ${y}_1$, $\ldots$, ${y}_{d}$ and $\mathrm{Vol}$ indicates the oriented volume of such simplex. Note that ${y}_0={0}$ in Eq.~\ref{GeomId}. Expanding the left hand side of Eq.~\ref{Identity}, we obtain
\begin{equation}\label{X1}
\alpha_d \int_{\Omega_{\Ll_0}}{\rm d}\bar \PM(\Ll) \ \sum_{w \in \Ll}\sum_{j=1}^{d+1}(-1)^{j+1} \mathrm{Vol}\big[{\bm 0},\widehat{y_1-w},\ldots, \widehat{\underline{y_j-w}},\ldots,\widehat{y_{d+1}-w}\big],
\end{equation}
where the underline means the entry is omitted. It is convenient to translate the simplexes and move the first vertex to a proper place to work with the simplexes
\begin{equation}
\mathfrak{S}_j(w)=[w+\widehat{y_1-w},\ldots, w,\ldots,w+\widehat{y_{d+1}-w}],
\end{equation}
where $w$ is located at the $j$-th position. We will also denote by $\mathfrak{S}$ the simplex
\begin{equation}
\mathfrak{S}=[y_1, \ldots, y_{d+1}],
\end{equation}
where we recall that $y_{d+1}$ coincides with the origin. To summarize, 
\begin{equation}\label{Intermetzo}
\begin{aligned}
\int_{\Omega_{\Ll_0}}{\rm d}\bar \PM(\Ll) \sum_{w \in \Ll} \ \mathrm{tr}_\gamma \Big (\gamma_0 \prod_{i=1}^{d} & \gamma \cdot \big(\widehat{ {y}_i-{w}}-\widehat{ {y}_{i+1}-{w}}\big )\Big )  \\
&  \quad =\alpha_d \int_{\Omega_{\Ll_0}}{\rm d}\bar \PM(\Ll) \sum_{w \in \Ll}\sum_{j=1}^{d+1} \mathrm{Vol}\{\mathfrak{S}_j({w})\}.
\end{aligned}
\end{equation}
The factor $(-1)^{j+1}$ disappeared because we changed the order of the vertices and note that, if $w$ is located inside $\mathfrak{S}$, then the orientations of the simplexes $\mathfrak{S}_j(w)$ coincide with that of $\mathfrak{S}$, because the former can be continuously deformed into the latter while keeping the volume finite (see Fig.~\ref{Fig:Fig5}(a) for a demonstration).

\begin{figure}
\center
  \includegraphics[width=12cm]{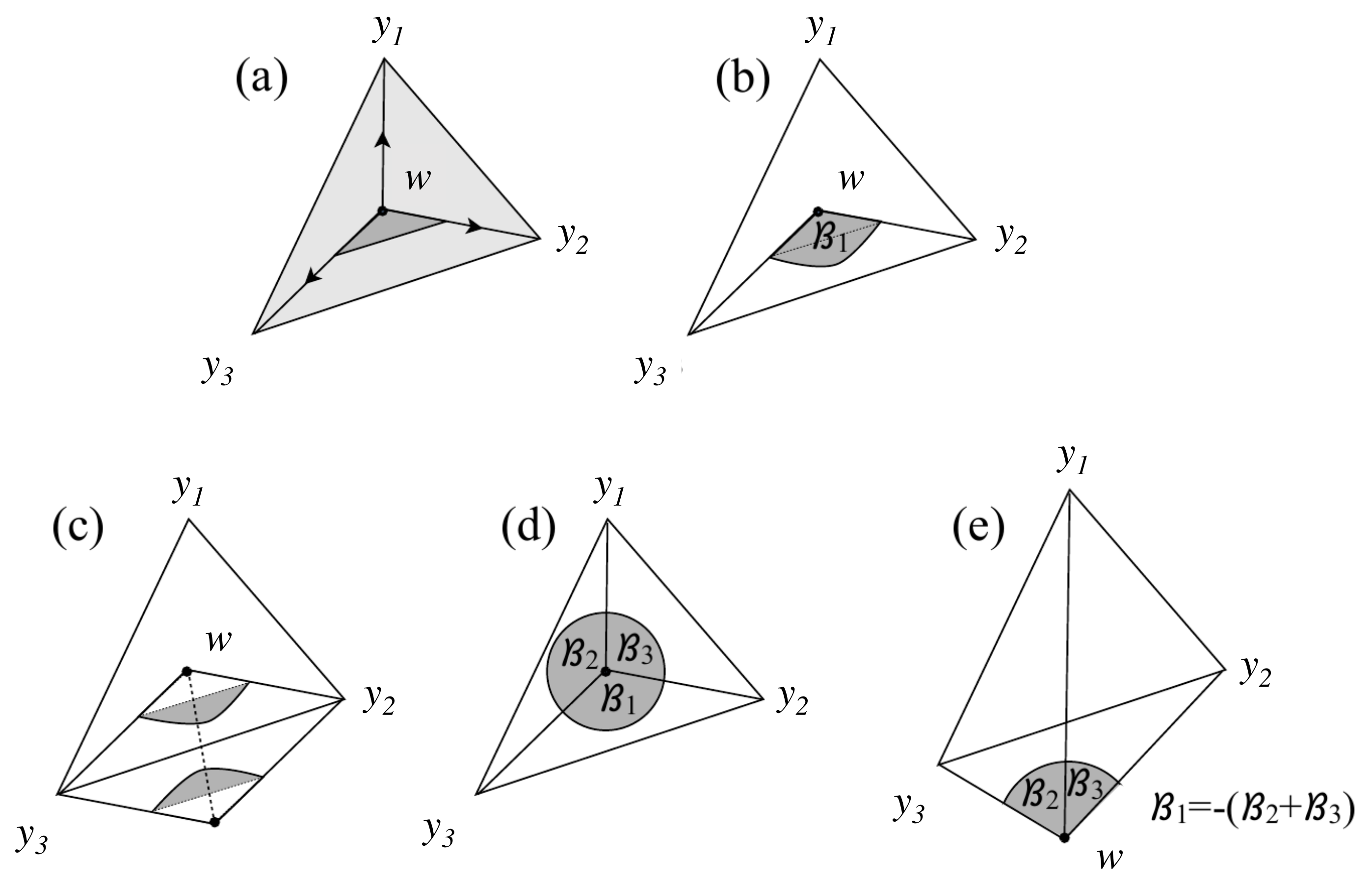}\\
  \caption{\small (a) Illustration of the simplexes $\mathfrak{S}=(y_1,y_2,y_3)$ (light gray) and $\mathfrak{S}_1(w)=(w,w+\widehat{y_2-w},w+\widehat{y_3-w})$ (darker gray), together with the interpolation process that takes $\mathfrak{S}_1(w)$ into $\mathfrak{S}$. (b) Illustration of the ball sector $B_1^1(w)$ (labeled as ${\Bb}_1$ in the figure) corresponding to the simplex $\mathfrak{S}_1(w)$. (c) Illustration of the inversion operation of $w$ relative to the center of the segment $(y_2,y_3)$. The volume of $B_1^1(w)-\mathfrak{S}_1(w)$ (shaded in gray) changes sign under this operation. (d) The ball sectors $B^1_1(w)$, $B_1^2(w)$ and $B_1^3(w)$(denoted as $\Bb_j$ in the figure) have same orientation and they add up to the full unit disk when $w$ is inside $\mathfrak{S}$. (e) When $w$ is outside of $\mathfrak{S}$, the ball sector $B_1^1(w)$ has the opposite orientation of $B_1^2(w)$ and $B_1^3(w)$ and, as a consequence, they add up to zero.}
  \label{Fig:Fig5}
\end{figure} 

Excepting $w$, the vertices of the simplexes $\mathfrak{S}_j(w)$ are all located on the unit sphere centered at $w$. As such, for any given simplex $\mathfrak{S}_j(w)$, the facets stemming from $w$ define a sector of the unit ball $B_1(w)$. This sector will be denoted by $B_1^j(w)$. This construction is illustrated in Fig.~\ref{Fig:Fig5}(b) for the 2-dimensional case. The orientations of $B_1^j(w)$'s are considered to be the same as those of $\mathfrak{S}_j(w)$'s. Now, an important observation is that
\begin{equation}\label{Asymptotic}
\mathrm{Vol}\{\mathfrak{S}_j(w)\}-\mathrm{Vol}\{ B^j_1(w)\} \sim |w|^{-(d+1)}
\end{equation}
asymptotically as $|w|\rightarrow \infty$. Consequently, we can write
\begin{eqnarray}
\int_{\Omega_{\Ll_0}}{\rm d}\bar \PM(\Ll) \sum_{w \in \Ll} \sum_{j=1}^{d+1} \mathrm{Vol}\{ \mathfrak{S}_j ({w})\}= \int_{\Omega_{\Ll_0}}{\rm d}\PM(\Ll) \sum_{w \in \Ll} \sum_{j=1}^{d+1} \mathrm{Vol}\{ B_1^j (w)\}\\
\indent +\sum_{j=1}^{d+1} \int_{\Omega_{\Ll_0}}{\rm d}\bar \PM(\Ll) \sum_{w \in \Ll} \big ( \mathrm{Vol}\{ \mathfrak{S}_j ({w})\}-\mathrm{Vol}\{ B_1^j ({w}) \} \big ), \nonumber
\end{eqnarray}
with the integrals in the second line being absolutely convergent. In fact, 
\begin{equation}\label{Fact1}
\int_{\Omega_{\Ll_0}}{\rm d} \bar \PM(\Ll) \sum_{w \in \Ll}\big ( \mathrm{Vol}\{ \mathfrak{S}_j (w)\}-\mathrm{Vol}\{ B_1^j (w) \} \big )=0
\end{equation}
for all $j=\overline{1,d+1}$, because, given our assumption, the integrand is odd under the inversion of $w$ relative to the center of the facet $\{ x_1,\ldots,\underline{x_j},\ldots,x_{d+1}\}$ of the simplex $\mathfrak{S}$ and these two cases are sampled with equal weight. This fact is illustrated in Fig.~\ref{Fig:Fig5}(c) for the 2-dimensional case. As for the remaining term,
\begin{equation}\label{Fact2}
 \sum_{j=1}^{d+1} \mathrm{Vol}\{ B_1^j (w)\}=\left \{
 \begin{array}{ll}
 \mathrm{Vol}\big(B_1(w)\big ) & \ \mbox{if} \ w \ {\rm inside} \ {\mathfrak S},\medskip \\
 0 & \ \mbox{if} \ w \ \mbox{outside} \ \mathfrak{S}.
 \end{array}
 \right .
 \end{equation}
This is illustrated in Figs.~\ref{Fig:Fig5}(d-e) for the 2-dimensional case. Then the right side of Eq.~\eqref{Intermetzo} reduces to
\begin{equation}
\begin{aligned}
\alpha_d \mathrm{Vol}\big(B_1(0)\big) \int_{\Omega_{\Ll_0}}{\rm d}\PM(\Ll) \ |\Ll \cap {\mathfrak S}|  = \alpha_d \mathrm{Vol}\big(B_1(0)\big)  |{\mathfrak S}| .
\end{aligned}
\end{equation}
We write the volume of the simplex $\mathfrak{S}$ as a determinant, in which case
\begin{equation}
\alpha_d \mathrm{Vol}\big(B_1(0)\big)  {\rm Vol}({\mathfrak S}) = \frac{(2\imath \pi)^\frac{d}{2}}{(d/2)!}{\rm Det}\big (y_1,y_2,\ldots,y_d\big )
\end{equation}
and the statement follows.
\end{proof}

\subsection{Concluding remarks and outlook} The main difficulty in extending the above proof to the cases with $N >1$ comes from the fact that the center of mass of the subsets live on a different, much finer lattice than $\Ll$. A more difficult task is to map the range of the pairings and this is where our efforts will go in the near future. For this, we need to develop new tools to compute the K-theories of the bi-equivariant subalgebra $C^\ast_{r,\Ss_N}(\widehat \Gg_{\Ll_0}^{(N)},\CM)$. In the meantime, we plan to examine the problem numerically.

\bibliographystyle{amsplain}

\end{document}